\documentclass[12pt]{article} 
\usepackage[colorlinks, citecolor={blue}, urlcolor={blue}]{hyperref}
\usepackage{url}
\usepackage{amsfonts,amscd,amssymb}
\usepackage{amsthm, amsmath}
\usepackage{natbib}
\usepackage{algorithm, algorithmicx, algpseudocode}
\usepackage{bm}
\usepackage{bbm} 
\usepackage[table]{xcolor}
\usepackage{soul}
\usepackage{verbatim}
\usepackage{graphicx}
\graphicspath{{sBMDS_figures/}}
\usepackage{setspace}
\usepackage[margin=1in]{geometry}
\usepackage{enumitem}
\usepackage{listings}
\usepackage{tikz}
\usetikzlibrary{shapes.misc}
\usepackage{etoolbox}
\usepackage{appendix}
\usepackage[format=plain,
labelfont={it},
textfont=it]{caption}
\usepackage{subcaption}
\usepackage{wrapfig}
\usepackage{xr}
\usepackage{booktabs}
\usepackage{multirow}
\usepackage{authblk}
\usepackage{mathbbol}
\usepackage{thmtools}

\usetikzlibrary{matrix}
\usetikzlibrary{backgrounds}
\usetikzlibrary{calc}
\usetikzlibrary{arrows,shapes}
\usetikzlibrary{decorations}
\usetikzlibrary{decorations.pathmorphing}
\usetikzlibrary{fit}
\usetikzlibrary{decorations.pathreplacing}
\usetikzlibrary{shapes.misc}
\usetikzlibrary{shapes.geometric}

\newtoggle{quickdraw}

\definecolor{lightgrey}{rgb}{0.9,0.9,0.9}
\definecolor{darkgreen}{rgb}{0,0.3,0}

\definecolorset{rgb}{}{}{darkred,0.8,0,0;darkgreen,0,0.5,0;darkblue,0,0,0.5}

\newtheorem{thm}{Theorem}
\newtheorem{lemma}{Lemma}

\newtheorem{remark}{Remark}

\newtheorem{assumption}{Assumption}

\newcommand{\RR}{\mathbb R}

\newcommand{\mb}{\mathbf}

\renewcommand{\vec}[1]{\mathbf{#1}}

\newcommand{\latentData}{\vec{X}}
\newcommand{\latentdata}{\vec{x}}

\newcommand{\phylogeny}{{\cal G}}
\newcommand{\tree}{\phylogeny}

\newcommand{\distanceMatrix}{\boldsymbol{\Delta}}

\newcommand{\traitVariance}{\mathbf{\Sigma}}

\newcommand{\mdsSD}{\sigma}
\newcommand{\mdsVariance}{\mdsSD^2}

\definecolor{trevorblue}{rgb}{0.330, 0.484, 0.828}
\definecolor{trevoryellow}{rgb}{0.829, 0.680, 0.306}

\allowdisplaybreaks

\title{Sparse Bayesian multidimensional scaling(s)}

\author[1]{Ami Sheth}
\author[2]{Aaron Smith}
\author[1]{Andrew J.~Holbrook}
\affil[1]{Department of Biostatistics, University of California, Los Angeles}
\affil[2]{Department of Mathematics and Statistics, University of Ottawa}

\begin{document}

\maketitle

\begin{abstract}
Bayesian multidimensional scaling (BMDS) is a probabilistic dimension reduction tool that allows one to model and visualize data consisting of dissimilarities between pairs of objects. Although BMDS has proven useful within, e.g., Bayesian phylogenetic inference, its likelihood and gradient calculations require a burdensome $\mathcal{O}(N^2)$ floating-point operations, where $N$ is the number of data points. Thus, BMDS becomes impractical as $N$ grows large. We propose and compare two sparse versions of BMDS (sBMDS) that apply log-likelihood and gradient computations to subsets of the observed dissimilarity matrix data. Landmark sBMDS (L-sBMDS) extracts columns, while banded sBMDS (B-sBMDS) extracts diagonals of the data. These sparse variants let one specify a time complexity between $N^2$ and $N$. Under simplified settings, we prove posterior consistency for subsampled distance matrices. Through simulations, we examine the accuracy and computational efficiency across all models using both the Metropolis-Hastings and Hamiltonian Monte Carlo algorithms. We observe approximately 3-fold, 10-fold and 40-fold speedups with negligible loss of accuracy, when applying the sBMDS likelihoods and gradients to 500, 1,000 and 5,000 data points with 50 bands (landmarks); these speedups only increase with the size of data considered. Finally, we apply the sBMDS variants to: 1) the phylogeographic modeling of multiple influenza subtypes to better understand how these strains spread through global air transportation networks and 2) the clustering of ArXiv manuscripts based on low-dimensional representations of article abstracts. In the first application, sBMDS contributes to holistic uncertainty quantification within a larger Bayesian hierarchical model. In the second, sBMDS provides uncertainty quantification for a downstream modeling task.
\end{abstract}

\section{Introduction}\label{sec:intro}

Multidimensional scaling (MDS) is a dimension reduction technique that maps pairwise dissimilarity measurements corresponding to a set of $N$ objects to a configuration of $N$ points within a low-dimensional Euclidean space \citep{torgerson1952cmds}. Classical MDS uses the spectral decomposition of a doubly centered matrix derived from the observed dissimilarity matrix to calculate the objects' coordinates. While classical MDS serves as a valuable data visualization tool, probabilistic extensions further enable uncertainty quantification in the context of Bayesian hierarchical models. \cite{ohraftery2001bmds} propose a Bayesian framework for MDS (BMDS) under the assumption that the observed dissimilarities follow independent truncated normal probability density functions (PDFs). BMDS facilitates Bayesian inference of object configurations in a manner that is robust to violations of the Euclidean model assumption and dimension misspecifications \citep{ohraftery2001bmds}. The key benefits of a Bayesian approach to MDS are that it provides uncertainty quantification for the projection itself and conditional distributions that can be easily integrated with other probability models, enabling fully model-based approaches to analyzing dissimilarity data. For example, one may incorporate BMDS into hierarchical modeling frameworks for Bayesian phylogeography \citep{bedford2014, holbrook2021bigbmds, li2023}, clustering \citep{BMDScluster}, and variable selection \citep{lin2019}.

Bayesian phylogeography uses molecular data from species such as viruses, bacteria or pathogens to probabilistically model their evolution over both time and space \citep{lemey2009}. For instance, one can reconstruct viral dispersion patterns to better understand the way viruses spread within and between human populations. The incorporation of BMDS within Bayesian phylogeography allows one to place dissimilarity data between species into a low-dimensional spatial representation while also considering their evolutionary dynamics from genetic data. \cite{bedford2014} simultaneously characterize antigenic and genetic patterns of influenza by combining BMDS with an evolutionary diffusion process on the latent strain locations. They apply BMDS on hemmagglutination inhibition assay data to place the subtypes on a low-dimensional antigenic map. \cite{holbrook2021bigbmds} implement a similar Bayesian phylogenetic MDS model, but perform phylogeographic inference on pairwise distances arising from air traffic data. Additionally, \cite{li2023} use phylogenetic BMDS on pairwise distances stemming from hepaciviruses to infer the viral locations in a lower dimensional geographic and host space.

Unfortunately, BMDS is difficult to scale to big data settings; computing the BMDS log-likelihood and gradient each have $\mathcal{O}(N^2)$ complexity. \cite{bedford2014} partially circumvent this problem by assuming that the observed data follow non-truncated Gaussian distributions, thereby avoiding the costly floating-point operations necessary to evaluate the Gaussian cumulative density function (CDF) in the truncated normal PDFs (\ref{eq:Fll}). However, there are benefits to using the truncated normal distribution: it appropriately accounts for non-negative dissimilarities, and its variance term is always less than that of its corresponding non-truncated normal distribution, resulting in more precise posterior inference. \cite{holbrook2021bigbmds} mitigate BMDS's computational burden through massive parallelization using multi-core central processing units, vectorization and graphic processing units. They obtain substantial performance gains, but parallelization requires expensive hardware. In either case, these models still scale quadratically in the number of objects. We therefore develop a framework that reduces the time complexity to $\mathcal{O}(N)$ by inducing sparsity on the observed dissimilarity matrix. We perform experiments with simulated data and show that our sparse versions of BMDS (sBMDS) obtain significant speedups while preserving inferential accuracy. We then illustrate how one may use sBMDS within a larger hierarchical model by extending the types of phylogeographical models mentioned above under sparse assumptions, implementing sBMDS phylogenetic frameworks on air traffic data to analyze the geographic spread of four influenza subtypes. Additionally, we apply sBMDS to a collection of ArXiv paper abstracts and perform post-hoc clustering on posterior samples of the low-dimensional embeddings. We adopt a ``bagged estimator" approach to propagate uncertainty quantification from sBMDS.

In the following, we present two versions of sparse BMDS and prove that under simplistic conditions, the posterior latent locations are consistent for subsampled dissimilarity matrices (Section \ref{sec:methods}). In Section \ref{sec:Results}, we evaluate the empirical accuracy, sensitivity to model misspecification and computational performance of both methods. We apply sBMDS to the phylogeographic modeling of influenza variants and verify that we obtain similar migration rate estimates for both full and sparse BMDS models (Section \ref{sec:fluapp}). In Section \ref{sec:ERapp}, we apply sBMDS to a dataset representing ArXiv paper abstracts and recover the posterior probability that two manuscripts belong to the same subject-matter category. We conclude by summarizing our findings and discussing future research directions (Section \ref{sec:discussion}).

\section{Methods} \label{sec:methods}

\subsection{Bayesian multidimensional scaling}

Bayesian multidimensional scaling (BMDS) \citep{ohraftery2001bmds} models a set of $N$ objects' locations as latent variables in low-dimensional space under the assumption that the observed dissimilarity measures follow a prescribed joint probability distribution. To set notation: for $A \subset \RR$, let $N_{A}(\mu,\sigma^{2})$ denote the Gaussian distribution truncated to $A$; for $k \in \mathbb{N}$, let $[k] = \{1,2,\ldots,k\}$. Within BMDS, each observed dissimilarity measure $\delta_{nn'}$ is the posited latent measure $\delta_{nn'}^*$ plus a truncated Gaussian error: 
\begin{align} \label{eq:MDSmodel}
    \delta_{nn'} \sim N_{(0, \infty)}(\delta_{nn'}^*, \sigma^2), \ n \ne n', \ n, n' \in [N],
\end{align} where $\delta_{nn'}^* = \sqrt{\sum_{d = 1}^D (x_{nd} - x_{n'd})^2}$ is the Euclidean distance between latent locations $\mathbf{x}_n, \,\mathbf{x}_{n'}$$\in \RR^D$, and $N(\cdot, \cdot)$ represents the normal distribution. These assumptions yield the log-likelihood function
\begin{align}
    \ell(\distanceMatrix, \sigma^2) = -\frac{m}{2} \log(2 \pi \sigma^2) - \sum_{n < n'} \biggr[ \frac{(\delta_{nn'} - \delta_{nn'}^*)^2}{2\sigma^2} + \log \Phi \biggr( \frac{\delta_{nn'}^*}{\sigma} \biggr) \biggr], \label{eq:Fll}
\end{align}    
where $\distanceMatrix = \{\delta_{n n'}\}$ is the symmetric $N \times N$ matrix of observed dissimilarities, $m = N(N-1)/2$ is the number of dissimilarities, and $\Phi(\cdot)$ is the standard normal CDF.  

Many MCMC algorithms, e.g., Hamiltonian Monte Carlo (HMC) (Section \ref{sec:MCMCInf}) and Metropolis-adjusted Langevin algorithm (MALA), use evaluations of gradients for efficient state space exploration. For this model, we take the first derivative of the log-likelihood function \eqref{eq:Fll} with respect to a single row $\latentdata_n$ of $\latentData$, the $N \times D$ matrix of unknown object coordinates to obtain the log-likelihood gradient function
\begin{align}
     \nabla{\latentdata_n}  \ell(\distanceMatrix, \sigma^2) = -\sum_{\substack{n' \in [N], \\ n' \ne n}} \biggr[ \biggr( \frac{(\delta_{nn'}^* - \delta_{nn'})}{\sigma^2} + \frac{\phi(\delta_{nn'}^*/\sigma)}{\sigma \Phi(\delta_{nn'}^*/\sigma)} \biggr) \ \frac{(\mathbf{x}_n - \mathbf{x}_{n'})}{\delta_{nn'}^*} \biggr] \equiv -\sum_{\substack{n' \in [N], \\ n' \ne n}} \mathbf{r}_{nn'}. \label{eq:Fgrad}
\end{align}
Here $\phi(\cdot)$ is the PDF of a standard normal variate, and $\mathbf{r}_{nn'}$ is the contribution of the $n'$th location to the gradient with respect to the $n$th location.

The BMDS log-likelihood \eqref{eq:Fll} and gradient \eqref{eq:Fgrad} both involve summing $\binom{N}{2}$ terms and require $\mathcal{O}(N^2)$ floating point operations. Given the large number of calculations needed, they become computationally cumbersome as the number of objects grows large. Therefore, we propose using a small subset of the data for likelihood and gradient evaluations, namely the sparse BMDS methods (sBMDS). 

\subsection{Sparse likelihoods and their gradients}
For each item $n$, let $J_{n, N} \subset [N] \setminus \{n\}$ be an index set. We consider sparse coupling approaches resulting in log-likelihoods and log-likelihood gradients of the form
\begin{align}
    \ell_{sp}(\distanceMatrix, \sigma^2) = - \sum_{n = 1}^N \sum_{\substack{n' \in J_{n, N}, \\ n' > n}} \biggr[\frac{1}{2} \log(2 \pi \sigma^2) + \frac{(\delta_{nn'} - \delta_{nn'}^*)^2}{2\sigma^2} + \log \Phi \biggr( \frac{\delta_{nn'}^*}{\sigma} \biggr) \biggr], \label{eq:Sll}
\end{align} 
and 
\begin{align}
    \nabla{\latentdata_n}  \ell_{sp}(\distanceMatrix, \sigma^2) = - \sum_{n' \in J_{n, N}} \mathbf{r}_{nn'}. \label{eq:Sgrad}
\end{align} 
We reduce the computational complexity of BMDS by including a small subset of couplings $J_{n, N}$ per object $n$, where $|J_{n, N}| \ll N$. Here, we discuss two possible strategies for choosing $J_{n, N}$. By a slight abuse of notation, we use $[a, b]$ to refer to a closed interval of either reals or integers, where the appropriate set should be obvious from context. The first option is to extract $B \in [N - 1]$ off-diagonal bands of the observed dissimilarity matrix such that $J_{n, N} = [\max(1, n - B), \min(N, n + B)] \backslash \{n\}$ for all $n$. The second approach is to choose $L \in [N]$ objects called ``landmarks" and select each landmark's dissimilarities from the remaining $N - 1$ objects, e.g., $J_{n, N} = [N] \backslash \{n\}$ for $n \in [L]$ and $J_{n, N} = [L]$ for $n \not\in [L]$. Essentially, this strategy retains a rectangular subset of the observed dissimilarity matrix by extracting $L$ columns (rows) of the data. We note there is no loss of generality in taking the first $n$ indices as landmarks rather than an arbitrary set because one can relabel the object indices without affecting the learned geometry. We refer to the first method as banded sBMDS (B-sBMDS) and the second method as landmark sBMDS (L-sBMDS). Alternative strategies for selecting index sets are possible, provided they satisfy Assumption \ref{GraphAssumption} (Section \ref{sec:postcons}). For instance, our B-sBMDS model assumes distances are measured on the real line and bands are defined as a contiguous interval. However, one could explore other forms of banded matrices, e.g., by selecting any set of entries in the distance matrix.

To highlight the difference, we consider a simplified scenario in which the number of objects is five, the latent dimension is two, the BMDS error variance $\mdsVariance$ is 0.25, and the observed dissimilarities are equal to the latent dissimilarity measures $(\delta_{nn'} = \delta^*_{nn'})$. Given the distance and location matrices $$\boldsymbol{\Delta} = \left[ {\begin{array}{ccccc} 
		0.00 & 1.35 & 2.53 & 0.99 & 1.85 \\
		1.35 & 0.00 & 1.54 & 0.76 & 0.50 \\
		2.53 & 1.54 & 0.00 & 1.54 & 1.26 \\
		0.99 & 0.76 & 1.54 & 0.00 & 1.12 \\
		1.85 & 0.50 & 1.26 & 1.12 & 0.00 \\
\end{array}}  \right], \quad \mathbf{X}  = \left[ {\begin{array}{cc} 
0.59 & 0.71  \\
-0.11 & -0.45  \\
0.61 & -1.82  \\
0.63 & -0.28 \\
-0.28 & -0.92  \\
\end{array}}  \right],$$
we compare the sBMDS log-likelihood (Table \ref{tab:sBMDS_llcomp}) and gradient (Table \ref{tab:sBMDS_gradcomp}) calculated from couplings defined by B-sBMDS versus L-sBMDS.

\begin{table}
\begin{center}
        \scalebox{0.9}{
    	\begin{tabular}{c | c | c }
    		\multicolumn{3}{c}{Pairs ($n, n'$)} \\ 
    		 B/L & B-sBMDS & L-sBMDS \\
    		\hline
    		1 & (1, 2); (2, 3); (3, 4); (4, 5)&  (1, 2); (1, 3); (1, 4); (1, 5) \\ 
    		2 & + (1, 3); (2, 4); (3, 5) & + (2, 3); (2, 4); (2, 5) \\  
    		3 & + (1, 4); (2, 5) & + (3, 4); (3, 5)  \\
    		4 &  + (1, 5)& + (4, 5) \\
    	\end{tabular}
        }
\quad
        \scalebox{0.9}{
    	\begin{tabular}{c | c | c }
    		\multicolumn{3}{c}{Log-likelihood values} \\ 
    		B/L & B-sBMDS & L-sBMDS \\
    		\hline
    		1 & -0.885 &  -0.875 \\ 
    		2 & -1.490 & -1.311 \\  
    		3 & -1.743 & -1.756  \\
    		4 & -1.969 & -1.969 \\
    	\end{tabular}
        }
\end{center}
\caption{We extract the $(n, n')$ pair from the off-diagonals of the observed and latent dissimilarity matrices for banded sBMDS (B-sBMDS) versus the columns for landmark sBMDS (L-sBMDS). $B/L$ refers to the number of bands ($B$) or landmarks ($L$), and the $+$ symbol indicates all couplings above are also included. The table on the right shows the calculated log-likelihoods as the number of bands/landmarks increases. Importantly, the bottom log-likelihoods are equal for both sBMDS variants and correspond to the full BMDS log-likelihood.}
\label{tab:sBMDS_llcomp}
\end{table}

\begin{table}
\begin{center}
	\doublespacing
	\scalebox{0.65}{
			\begin{tabular}{c p{2.5cm} p{3cm} p{2.5cm} p{2.5cm} | p{2.5cm} p{2.5cm} p{2.5cm} p{2.5cm}}
				\multicolumn{9}{c}{Banded sBMDS} \\ 
				& Pairs ($n, n'$) &  & & & Gradient \\
				\hline
				& 1 band & 2 bands & 3 bands & 4 bands & 1 band & 2 bands & 3 bands & 4 bands \\ 
				$\mathbf{x}_1$ & (1, 2) & + (1, 3) & + (1, 4) & + (1, 5) & [-.010, .017] & [-.010, .018] & [-.005, .134] & [-.006, .135]\\ 
				$\mathbf{x}_2$ &  (2, 3); (2, 1) & + (2, 4) & + (2, 5) & & [.014, .011] & [.275, .074] & [.071, -.468] & [.071, -.468]\\ 
				$\mathbf{x}_3$ &  (3, 4); (3, 2) & + (3, 5); (3, 1) & & & [-.003, .013] & [-.026, .036] & [-.026, .036] & [-.026, .036]\\ 
				$\mathbf{x}_4$ &  (4, 5); (4, 3) & + (4, 2) & + (4, 1) & & [-.054, -.045]& [-.315, -.108] & [-.321, .009] & [-.321, .009]\\ 
				$\mathbf{x}_5$ & (5, 4) & + (5, 3) & + (5, 2) & + (5, 1) & [.054, .038] & [.077, .015] & [.281, .557] & [.281, .558]\\ 
			\end{tabular}
	}
 
	\scalebox{.65}{
		\begin{tabular}{c p{2.5cm} p{3cm} p{2.5cm} p{2.5cm} | p{2.5cm} p{2.5cm} p{2.5cm} p{2.5cm}}
			\multicolumn{9}{c}{Landmark sBMDS} \\ 
			& Pairs ($n, n'$) &  & & & Gradient \\
			\hline
			& 1 landmark & 2 landmarks & 3 landmarks & 4 landmarks & 1 landmark & 2 landmarks & 3 landmarks & 4 landmarks \\ 
			$\mathbf{x}_1$ & (1, 2 - 5) & & & & [-.006, .135] & [-.006, .135] &[-.006, .135] & [-.006, .135] \\ 
			$\mathbf{x}_2$ & (2, 1) & + (2, 3 - 5) & & & [.010, .017] & [.071, -.468] & [.071, -.468] & [.071, -.468] \\  
			$\mathbf{x}_3$ & (3, 1) & + (3, 2) & + (3, 4 - 5) & & [.000, .000] & [-.003, .006] & [-.026, .036] & [-.026, .036] \\
			$\mathbf{x}_4$ &  (4, 1) & + (4, 2) & + (4, 3) & + (4, 5) & [-.005, .117] & [-.266, .054] & [-.266, .047] & [-.321, .009]\\
			$\mathbf{x}_5$ & (5, 1) & + (5, 2) & + (5, 3) & + (5, 4) & [.000, .000] & [.204, .543] & [.227, .519] & [.281, .558]\\
		\end{tabular}
	}
\end{center}
\caption{We extract the $(n, n')$ pair from the observed and latent dissimilarity matrices to calculate the sBMDS gradients. On the left, the $-$ symbol, as in $(n, a-c)$, indicates pairs $(n, a); (n, b); (n, c)$. For example, pair $(3; 4-5)$ means we include both pair $(3, 4)$ and $(3, 5)$. The $+$ symbol indicates all couplings to the left are also included, and $[\cdot, \cdot]$ represents a vector. On the right is the gradient computed for each $\latentdata_n$ of $\latentData$ as function of the number of bands/landmarks. Extracting the entire column of a landmark point gives the full BMDS gradient in $\RR^D$ whereas banded sBMDS incrementally adds information to the row-wise gradients. Importantly, the rightmost gradients are equal for both sBMDS variants and correspond to the full BMDS gradient.}
\label{tab:sBMDS_gradcomp}
\end{table}

For B-sBMDS, the number of couplings is the number of elements in $B$ bands. The relationship between the number of bands and number of couplings $C$ is $C = \sum_{b = 1}^{B} (N - b)$. We add one less coupling for each additional band. When the number of bands equals $N - 1$, we return to the full BMDS case. Using a subset of the observed dissimilarity matrix reduces the burden of computing the BMDS likelihood and gradient to $\mathcal{O}(NB)$. Similar arguments hold for L-sBMDS, the likelihoods and gradients of which exhibit $\mathcal{O}(NL)$ time complexity. 

For classical MDS, an analogous strategy to L-sBMDS already exists. In MDS, the rate limiting step is the calculation of the top $D$ eigenvalues and eigenvectors from a $N \times N$ matrix. \cite{mdsLM2004} propose applying classical MDS to $L$ landmark points, e.g., an $L \times N$ submatrix of the observed dissimilarity matrix, and then following a distance-based triangulation procedure to determine the remaining object coordinates. L-sBMDS uses the concept of randomly selecting $L$ landmarks as well, but integrates them into the BMDS framework, allowing inference on the entire model. \cite{casecntrl2012} approximate the likelihood of their network data by taking a random subset of objects deemed to have no link, reducing the time complexity from $\mathcal{O}(N^2)$ to $\mathcal{O}(N)$. In the context of a very different network model, they incorporate an array of covariates to model the probability of a link between objects $n$ and $n'$ while our model is simpler, using no outside information to aid in determining locations in a latent space. 

\subsection{Posterior consistency} \label{sec:postcons}

For the following theoretical development, we consider the model
\begin{align}\label{eq:GMDSmodel}
    \delta_{nn'} \sim N_{(0, M)}(\delta_{nn'}^*, \sigma^2), \ n \ne n', \ n, n' \in [N],
\end{align} 
a generalization of (\ref{eq:MDSmodel}) insofar as M can be any number within the interval $(0,\infty)$. Let the latent locations $\latentData$ be sampled from a range of values in the interval $I$. The posterior density function of the unknown parameters ($\latentData, I, \sigma^2, M$) is proportional to $\mathcal{L}(\distanceMatrix| \latentData, \sigma^2, M)$, the BMDS likelihood function of model (\ref{eq:GMDSmodel}), and the priors put on each auxiliary parameter, e.g., 
\begin{align} \label{eq:posterior}
    p(\latentData, I, \sigma^2, M | \distanceMatrix) \propto  \mathcal{L}(\distanceMatrix, \sigma^2, M | \latentData) \times p(\latentData | I) \times p(I) \times p(\sigma^2) \times p(M).
\end{align}
The marginal posterior density function of $\latentData$ is 
\begin{align}\label{eq:postMarg}
    p(\latentData | \distanceMatrix) = \int p(\latentData, I, \sigma^2, M | \distanceMatrix) \, dI \, d\sigma^2 \, dM.
\end{align}

We examine the posterior consistency of subsampled dissimilarity matrices under simple conditions. Fixing some interval $I$, we sample points $x_{1}, \dots, x_{N} \stackrel{i.i.d.}{\sim} N_{I}(0, 1)$. Let $\Delta^*$ be the associated Euclidean distance matrix with entries $\delta^*_{nn'} = |x_n - x_{n'}|$ and $\delta_{nn'}$ be the truncated noisy observations of this matrix sampled from model (\ref{eq:GMDSmodel}). We set a prior on $I$, $M$ and $\sigma^2$ that has compact support and is bounded away from 0 and infinity on its support, e.g., \cite{ohraftery2001bmds, BMDScluster}. In addition, we fix in advance a collection of indices $J_{n, N} \subset [N] \backslash \{n\}$ of observations to keep for each object $n$, treating this choice as non-random in the following. Next, we make some assumptions about which observations are kept.

\begin{assumption} \label{GraphAssumption}
Fix $K \in \mathbb{N}$. Assume there exists a sequence $\{ \ell_{N}\}_{N \in \mathbb{N}}$ and a collection of partitions $\{G_{n, N}^{(k(n))}\}_{k=1}^{K}$ of $J_{n, N}$ with the following properties:
\begin{enumerate}
    \item For all $n \in [N]$ and $k \in [K]$, $|G_{n, N}^{(k(n))}| \geq \ell_{N}$. 
    \item Say $n,n' \in [N]$ are linked by an edge if there exists $k(n), k(n') \in [K]$ so that 
    \begin{align}
        | J_{n,N} \cap J_{n',N} \backslash (G_{n,N}^{(k(n))} \cup G_{n',N}^{(k(n'))})| \geq \ell_{N}.
    \end{align}
    Assume that the graph with these edges and vertex set $[N]$ is a connected graph.
    \item The sequence $\ell_{N}$ satisfies 
    \begin{align}
            \lim_{N \rightarrow \infty} \frac{\ell_{N}}{\log(N)^{2}} = \infty.
    \end{align}
\end{enumerate}
\end{assumption}

\begin{remark} \label{LAssumpt} 
We verify that Assumption \ref{GraphAssumption} holds for L-sBMDS given $L \geq (\ell_N + 1)K$ landmarks and $L \ll N$. We can think of $\ell_N$ as the number of retained entries in the sparsest row (up to a universal constant). Let $\ell_N = \lceil\frac{\sqrt{N}}{2}\rceil$, so that it satisfies part 3 of Assumption \ref{GraphAssumption}. For all $n \in [N], |J_{n, N}| \geq L$ and for objects $n, n' \in [N], |J_{n, N} \cap J_{n', N}| \geq L - 1$. Consider the partition $G_{n, N}^{(k(n))} = [\lfloor\frac{(k -1)L}{K}\rfloor + 1, \lfloor\frac{kL}{K}\rfloor] \backslash \{n\}$. Then for any object $n, n' \in [N]$, $G^{(k(n))}_{n, N} \cup G^{(k(n'))}_{n', N} = \{\lfloor\frac{(k - 1)L}{K}\rfloor + 1,...,\lfloor\frac{kL}{K}\rfloor\} \equiv G_N^{(k)}$ is independent of $n, n'$ and of size $\frac{L}{K}$, which satisfies part 1 of Assumption \ref{GraphAssumption}. Thus in the most conservative case, e.g., $n \in [L]$ and $n' \notin [L]$,
\begin{align*}
    |J_{n, N} \cap J_{n', N} \backslash (G^{(k(n))}_{n, N} \cup G^{(k(n'))}_{n', N})| = |\{1,...,n -1, n + 1,...,L\} \backslash G_N^{(k)}| \\ = (L - 1) - \frac{L}{K} = \frac{L(K-1)}{K} - 1.
\end{align*}
As a result, $|J_{n, N} \cap J_{n', N} \backslash (G^{(k(n))}_{n, N} \cup G^{(k(n'))}_{n', N})| \geq \ell_N$ when $L \geq (\ell_N + 1)K$ for $K > 1$. Notably, the graph in part 2 of Assumption \ref{GraphAssumption} is connected.
\end{remark}

\begin{remark} \label{BAssumpt} 
Similarly, we verify that Assumption \ref{GraphAssumption} holds for B-sBMDS given $B \geq 2 \ell_N + 1$ bands and $B \ll N$. Again, let $\ell_N = \lceil\frac{\sqrt{N}}{2}\rceil$, so that it satisfies part 3 of Assumption \ref{GraphAssumption}. Under B-sBMDS, the ends of a distance matrix have the fewest indices. To ensure that the size of each partition is at least $\ell_N$ at these boundaries, let the number of partitions be $K = \lfloor\frac{B}{\ell_N}\rfloor$, satisfying part 1 of Assumption \ref{GraphAssumption}. For all $n \in [N], |J_{n, N}| \geq B$ and for two consecutive objects $(n < n') \in [N], |J_{n, N} \cap J_{n', N}| \geq B - 1$. To remove the minimal number of common indices between two consecutive objects, let $G_{n, N}^{(k(n))} = \{n + 1,..,n + \ell_N \}$ so that $|G_{n, N}^{(k(n))}| = \ell_N$. Then, $G_{n, N}^{(k(n))} \cup G_{n', N}^{(k(n'))} = \{n', n' + 1, ..., n' + \ell_N\}$ and the cardinality of the intersection is $\ell_N + 1$. Finally, 
\begin{align*}
    |J_{n, N} \cap J_{n', N} \backslash (G_{n, N}^{(k(n))} \cup G_{n', N}^{(k(n'))})| \geq |\{n' + 1, ..., n + B\} \backslash \{n', n' + 1, ..., n' + \ell_N\}| \\ \geq |\{n' + \ell_N + 1,...,n + B\}| \geq (B - 1) - \ell_N
\end{align*}
because we remove $\ell_N$ elements from the union containing at least $B - 1$ elements. Thus when $B \geq 2\ell_N + 1$, $|J_{n, N} \cap J_{n', N} \backslash (G_{n, N}^{(k(n))} \cup G_{n', N}^{(k(n'))})| \geq \ell_N$, and we obtain a connected graph, fulfilling part 2 of Assumption \ref{GraphAssumption}. 



\end{remark}

Under Assumption \ref{GraphAssumption}, we have the following posterior consistency result.

\begin{thm} \label{ThmPostConst}
Fix $0 < \alpha <0.1$ and $K \in \mathbb{N}$. Let the sequences $\{J_{n, N}\}$, $\{G_{n, N}^{(k(n))}\}$ and $\{\ell_{N}\}$ satisfy Assumption \ref{GraphAssumption}. Let $\epsilon_{N} = \ell_{N}^{-0.5 + \alpha}$. Let $(x_{1}^{(N)}, \dots, x_{N}^{(N)}) \stackrel{i.i.d.}{\sim} N_{I}(0, 1)$ and let $\{ \delta_{nn'}^{(N)}\}_{1 \leq n < n' \leq N}$ be sampled from model \eqref{eq:GMDSmodel}. Finally, let $(\Tilde{x}_{1}^{(N)},\dots, \Tilde{x}_{N}^{(N)}) \sim p(\cdot | \{\delta_{nn'}^{(N)}\}_{n' \in J_{n, N}})$ be sampled from the associated marginal posterior distribution of the model. Then there exists $C > 0$ so that the event 
\begin{align}
    \{ \forall n \in [N], \, |x_{n}^{(N)} - \Tilde{x}_{n}^{(N)}| < C \epsilon_{N}\}
\end{align}
occurs asymptotically almost surely.
\end{thm}

\begin{proof}
    See  Appendix \ref{sec:Appendix A}.
\end{proof}

Theorem \ref{ThmPostConst} proves that we can achieve posterior consistency for latent locations estimated from a subsampled dissimilarity matrix as we are able to recover the estimated latent locations $\Tilde{x}_{n}^{(N)}$ up to an additive error of $\mathcal{O} (\epsilon_N)$ relative to the true latent locations $x_{n}^{(N)}$. We acknowledge that the biggest limitation of this proof is the assumption that we have one-dimensional latent objects. See Section \ref{SecHigherDimTheory} for a short discussion of how similar results may be obtained in fixed dimensions greater than 1. 

\subsection{Bayesian computation} \label{sec:MCMCInf}

Bayesian hierarchical models under the BMDS framework have previously been fit using MCMC algorithms such as Metropolis-Hastings (MH) \citep{Metropolis, MH1997, ohraftery2001bmds, bedford2014} and HMC \citep{neal2012, holbrook2021bigbmds}. In the following, we experiment with MH and HMC to perform posterior inference with the sBMDS models.  

Let $\theta$ be the random variable of interest and $\pi(\theta)$ the target distribution. Under MH, a new candidate $\theta^*$ is sampled from a proposal distribution centered at the value of the current iteration $s$, $q(\theta^* | \theta^{(s)})$. One then accepts the candidate with probability 
\begin{align}
    \alpha(\theta^* | \theta^{(s)}) = \min \biggr[ 1, \frac{\pi(\theta^*)q(\theta^{(s)} | \theta^*)}{\pi(\theta^{(s)})q(\theta^* | \theta^{(s)})} \biggr]. \label{eq:MHAR}
\end{align} 
In the BMDS model \eqref{eq:MDSmodel}, the parameters of interest are the latent locations $\latentData$ and the error variance $\mdsVariance$, and--within a larger Metropolis-within-Gibbs scheme--the target distributions of interest are their respective conditional posterior distributions.

For our MH-based experiments, we jointly draw each candidate object's latent location $\latentdata_n^*$ from the normal proposal distribution, $N(\latentdata_n^{(s)}, \tau^2)$, in which the proposal standard deviation $\tau$ is a tuning parameter. In practice, we find it beneficial to adjust $\tau$ in a manner that satisfies the diminishing adaptations criterion of \cite{ronsenthal2001}. Specifically, the acceptance ratio is the number of acceptances in a given sample bound. If the acceptance ratio exceeds the target acceptance ratio, we multiplicatively increase $\tau$ by $(1 + \min(0.01, 1/ \sqrt{s - 1}))$; otherwise we multiplicatively decrease $\tau$ by $(1 - \min(0.01, 1/ \sqrt{s - 1}))$. 

For BMDS and its sparse variants, the dimension of the state space grows with the number of objects. Because MH typically breaks down in high-dimensions, we also consider HMC to infer the latent locations. HMC allows one to generate a Markov chain with distant proposals that nonetheless have a high probability of acceptance. It combines a fictitious momentum variable, $\mb{P}$, along with a position variable to create a Hamiltonian system from which we compute the trajectories necessary for state space exploration. The position variable represents the parameters of the target distribution, so in the context of our model, we let the position variable be the latent locations $\latentData$. The Hamiltonian function is 
\begin{align}
    H(\latentData, \mb{P}) = U(\latentData) + K(\mb{P}) \label{eq:HMCfunc}
\end{align} where $U(\latentData)$ is the potential energy defined as the negative log target density, and $K(\mb{P})$ is the kinetic energy defined as $K(\mb{P}) = \text{tr}(\mb{P}^T \mb{P}) / 2$. The partial derivatives of the Hamiltonian dictate how $\mb{P}$ and $\latentData$ change over time $t$:
\begin{align}
    \frac{d\latentData}{dt} = \frac{\partial H(\latentData, \mb{P})}{\partial \mb{P}} = \mb{P}, \quad \frac{d\mb{P}}{dt} = \frac{-\partial H(\latentData, \mb{P})}{\partial \mb{\latentData}} = -\nabla_{\latentData} \ell(\distanceMatrix, \sigma^2). \label{eq:HMCpd}
\end{align} For computer implementation, these equations are discretized over time using some small stepsize $\epsilon$. We follow \cite{neal2012} and implement the leapfrog method to numerically integrate Hamilton's equations (\ref{eq:HMCpd}). We tune the stepsize in the same way we change the proposal standard deviation in the adaptive MH algorithm. To propose a new state, we sample an initial momentum variable $\mb{P}_0$ and numerically integrate Hamilton's equations with initial state, $(\latentData^{(s)}, \mb{P}_0)$. We then accept the proposed state, $(\latentData^*, \mb{P}^*)$, according to the Metropolis-Hastings-Green \citep{green1995, geyer2011} probability of
\small
\begin{equation}
        \min \biggr[ 1, \exp(-H(\latentData^*, \mb{P}^*) + H(\latentData^{(s)}, \mb{P_0})) \biggr] = \min \biggr[ 1, \exp(-U(\latentData^*) + U(\latentData^{(s)}) - K(\mb{P}^*) + K(\mb{P_0}) \biggr]. \label{eq:HMCAR}
\end{equation}
\normalsize

Measured on an iteration by iteration basis, HMC allows for faster exploration of state spaces, especially in higher dimensions, compared to MH \citep{neal2012, beskos2013}. However, HMC is computationally more expensive because it requires the gradient of the target function within every iteration of the leapfrog method. Recall that these gradient evaluations scale $\mathcal{O}(N^2)$ for BMDS. If we want to learn the BMDS error variance $\mdsVariance$ as well, we again follow the adaptive MH algorithm, drawing a candidate $\sigma^{2*}$ from a truncated normal proposal distribution with the current iteration's $\sigma^{2(s)}$ as the mean and a standard deviation with the same adaption scheme as described above. We account for the asymmetric proposal distribution within the MH acceptance probability (\ref{eq:MHAR}).

\section{Results} \label{sec:Results}

We explore the accuracy of full and sparse BMDS as well as the computational efficiency of all models in the context of the MH and HMC algorithms. The code for this project is available on Github (\url{https://github.com/andrewjholbrook/sparseBMDS}). For visualization, we use the \texttt{ggplot2} \citep{ggplot2} package in \texttt{R} \citep{Rcite}. 

\subsection{Simulation studies}

For a full Bayesian analysis, we put a D-dimensional multivariate normal distribution with mean $\mb{0}$ and diagonal covariance matrix $\mb{\Lambda}$ as the prior for $\latentdata_n$, independently for $n = 1,...,N$. The prior for the BMDS error variance $\mdsVariance$ is an inverse gamma with rate $a$ and shape $b$. One can define hyperpriors for $\mb{\Lambda}, a, b$, but we assume those parameters are fixed and known in this section. For our simulations, we set $\mb{\Lambda}$ equal to the identity $\mb{I}_2, a = 1$ and $b = 1$ so that $\latentdata_n \sim N(0, \mb{I}_2)$ and $\mdsVariance \sim IG(1,1)$. To create the observed dissimilarity matrix $\distanceMatrix = \{\delta_{nn'}\}$, we add i.i.d. noise using a truncated normal distribution with mean 0 and variance $\sigma^2_{true}$ to a ``true" distance matrix. For the ``true" distance matrix, $\distanceMatrix^{(true)} = \{\delta_{nn'}^{(true)}\}$, we generate a $N \times 2$ ``true" location matrix $\latentData$ from standard normal distributions and use $\latentData$ to calculate the Euclidean distance between pairs $(n, n')$.

\subsubsection{Accuracy} \label{sec:accuracy}

We test the accuracy of the sBMDS models by comparing the simulated ``true" dissimilarities to those obtained from HMC using the sBMDS posteriors and gradients. Given $S$ iterations, we calculate the mean of the mean squared error ($\overline{\mbox{MSE}}$) as $\overline{\mbox{MSE}} = \frac{1}{Sm} \sum_{s=1}^S \sum_{n \ne n'} (\delta_{nn'}^{*(s)} - \delta_{nn'}^{(true)})^2$ where $\delta_{nn'}^{*(s)}$ is the Euclidean distance calculated from the inferred locations of object $n$ and object $n'$ at iteration $s$, $\delta_{nn'}^{(true)}$ is the ``true" Euclidean distance, and $m = N(N-1)/2$ is the number of dissimilarities. We compare distances instead of locations because the locations are not identifiable under distance preserving transformations. For computational convenience, when the number of objects is greater than 1,000, we randomly sample 1,000 distances to calculate $\overline{\mbox{MSE}}$. We set $\sigma_{true}$ to either $0.1, 0.2, 0.3$ or $0.4$ to change noise levels and run 110,000 iterations, discarding the first 10,000 as burn-in and retaining every 100th iteration. We establish the initial conditions of the latent locations within HMC from classical MDS output. 

\begin{figure}
	\centering
	\includegraphics[width= 0.9 \textwidth]{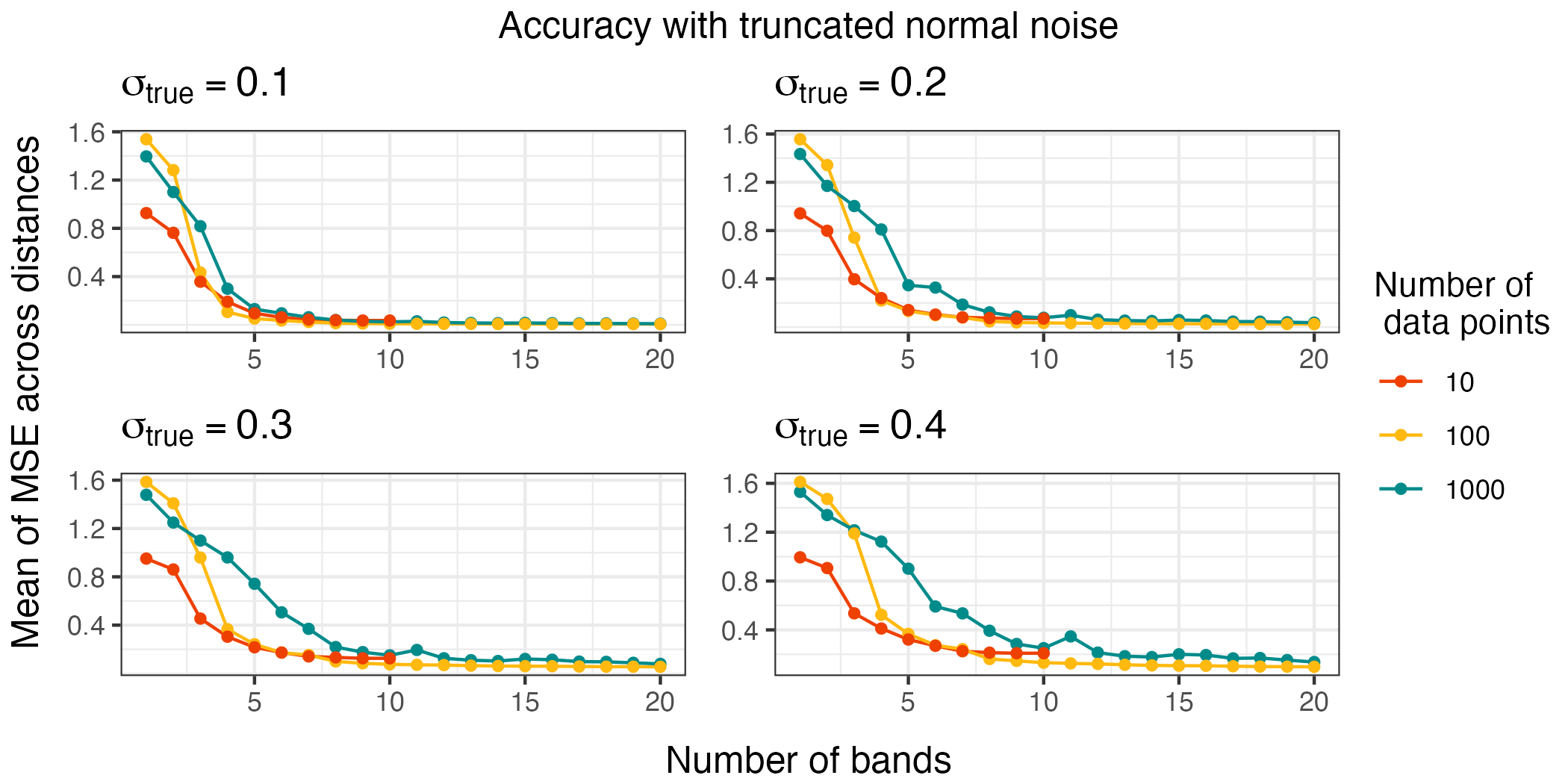}
	\caption{The mean of the mean squared error (MSE) across all distances using $1$ to $10$ bands for $10$ data points and $1$ to $20$ bands for $100$ and $1{,}000$ data points. We estimate Euclidean distances from the inferred locations obtained using an adaptive Hamiltonian Monte Carlo algorithm under banded sparse Bayesian multidimensional scaling (B-sBMDS). $\sigma^2_{true}$ is the variance component of the truncated normal noise centered at $0$ added to the ``true" distance matrix such that $\sigma_{true}$ corresponds to the BMDS error standard deviation $\sigma$.}
 \label{fig:mse_BAND}
\end{figure}

\begin{figure}
	\centering
	\includegraphics[width= 0.9 \textwidth]{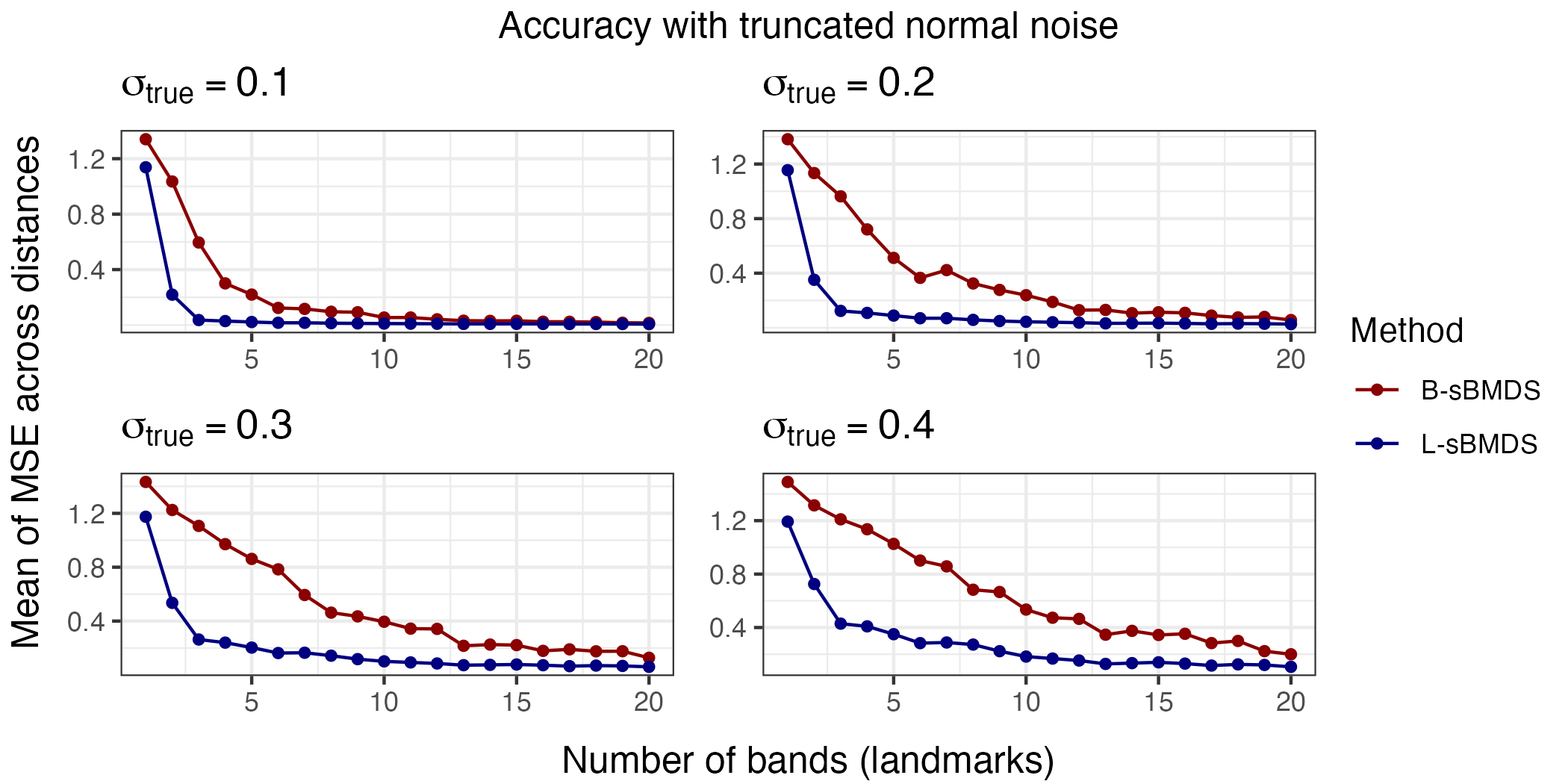}
	\caption{The mean of the mean squared error (MSE) across $1{,}000$ distances, randomly sampled from distance matrices with $10{,}000$ data points. We estimate Euclidean distances from the inferred locations obtained using an adaptive Hamiltonian Monte Carlo algorithm under both sparse Bayesian multidimensional scaling (sBMDS) variants, banded sBMDS (B-sBMDS) and landmark sBMDS (L-sBMDS) with $1$ to $20$ bands/landmarks. $\sigma^2_{true}$ is the variance component of the truncated normal noise centered at $0$ added to the ``true" distance matrix such that $\sigma_{true}$ corresponds to the BMDS error standard deviation $\sigma$.}
 \label{fig:mse_n10000}
\end{figure}

\begin{figure}
	\centering
	\includegraphics[width= 0.9 \textwidth]{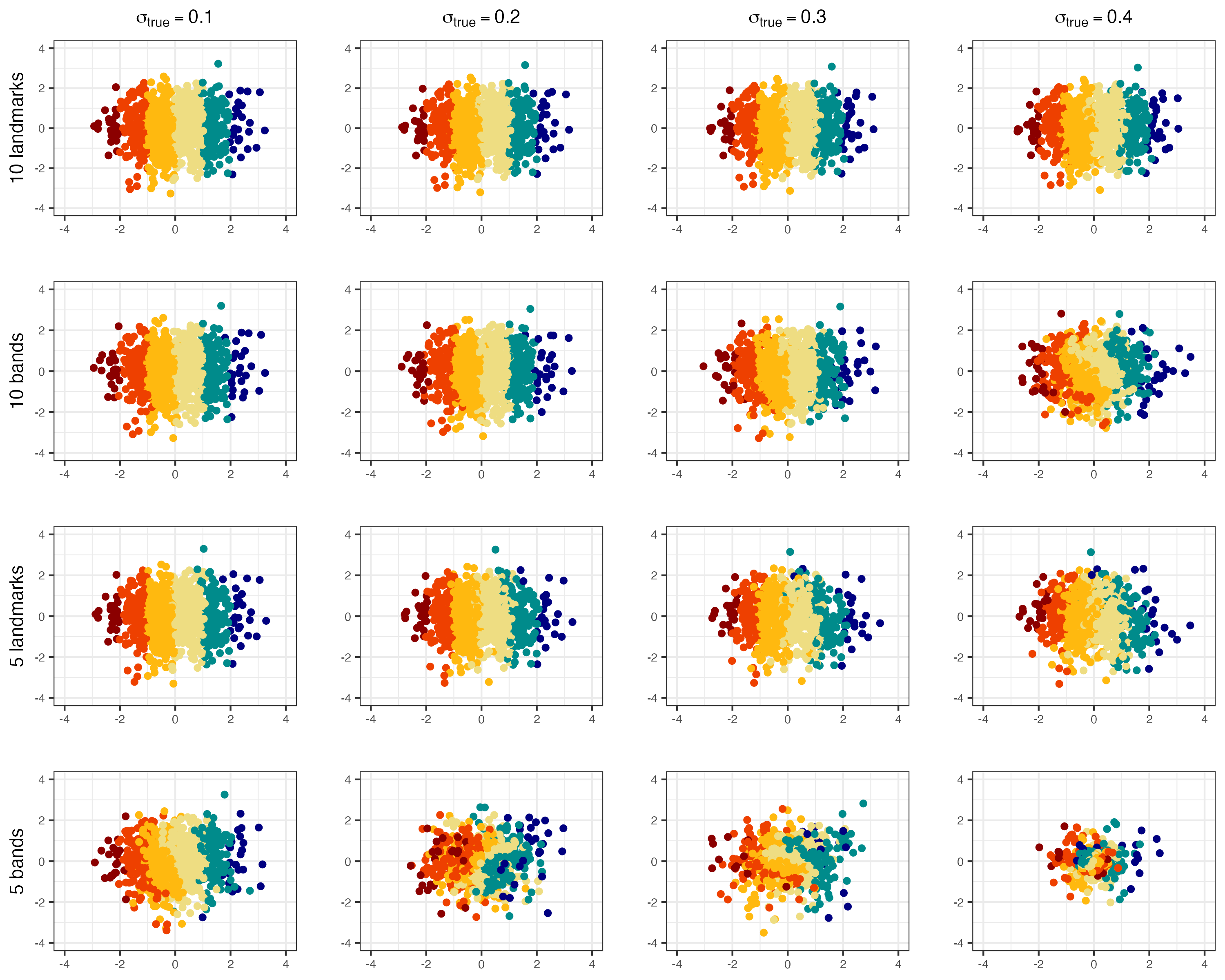}
	\caption{Procrustes aligned means of the inferred locations across $100{,}000$ iterations under the B-sBMDS and L-sBMDS frameworks when the number of bands/landmarks is ten and five. The number of data points is 1,000. We simulate the latent locations from a two-dimensional standard normal distribution and assign a color according to their x-coordinate. $\sigma^2_{true}$ is the variance component of the truncated normal noise centered at $0$ added to the ``true" distance matrix such that $\sigma_{true}$ corresponds to the BMDS error standard deviation $\sigma$.}
 \label{fig:gaussian.error}
\end{figure}

Figure \ref{fig:mse_BAND} plots $\overline{\mbox{MSE}}$ as function of the number of bands for data with 10, 100 and 1,000 data points at varying levels of noise (see Appendix \ref{sec:Appendix B}, Figure \ref{fig:mse_LM} for landmark results). Likewise, Figure \ref{fig:mse_n10000} plots $\overline{\mbox{MSE}}$ as function of the number of bands/landmarks for data with 10,000 data points under B-sBMDS and L-sBMDS at different noise levels. In both figures, all the plots have identifiable elbows, demonstrating that a small number of bands/landmarks is sufficient to achieve low error. While we need more bands for noisier data, the amount is still modest compared to the number of objects. Interestingly, we detect an elbow earlier for L-sBMDS than B-sBMDS; L-sBMDS recovers accurate pairwise relationships more efficiently than B-sBMDS. We visually see this difference in Figure \ref{fig:gaussian.error}. In this simulation, we generate 1,000 data points using the same sampling scheme and color-code the x-axis of the ``true" locations. After running 110,000 HMC samples, we plot the mean of the inferred latent locations from B-sBMDS and L-sBMDS using 5 and 10 bands/landmarks. From Figure \ref{fig:gaussian.error}, we observe that while L-sBMDS maintains the integrity of the latent locations, B-sBMDS rapidly loses its accuracy as noise increases for 10 bands and is no longer accurate for 5 bands.

\subsubsection{Sensitivity to model misspecification}

To observe how the sparse variants behave under model misspecifications, we explore two possible situations: 1) a mismatch between the true dimensionality and that specified by the scientist and 2) heavy-tailed, rather than truncated normal, noise. For case 1, we vary the dimension of the ``true" location matrix from 2 to 10 while fixing the embedding dimension to 2 and the truncated Gaussian noise variance $\sigma^2_{true}$ to 0.2. As expected, $\overline{\mbox{MSE}}$ decreases as the true underlying dimensionality approaches the embedding dimension. Full BMDS and B-sBMDS with 20 bands are more robust to dimension misspecification than classical MDS, and the accuracy for B-sBMDS with 20 bands closely matches that of full BMDS (Figure \ref{fig:HD.error}).

For case 2, we assume a correctly-specified-dimensional Euclidean space, but add i.i.d log-normal noise to the ``true" distance matrix. We bootstrap the $\overline{\mbox{MSE}}$ across all distances from 100 data points, a 100 times and plot the mean of $\overline{\mbox{MSE}}$ along with error bars representing $\pm$ the standard deviation of $\overline{\mbox{MSE}}$. Figure \ref{fig:LT.error} demonstrates that, even with heavy-tailed data, both sparse variants achieve comparable $\overline{\mbox{MSE}}$s to full BMDS's at a low number of bands/landmarks. In both cases, we observe that B-sBMDS seems to be less sensitive to model misspecification than L-sBMDS.

\begin{figure}
	\centering
	\includegraphics[width= 0.9 \textwidth]{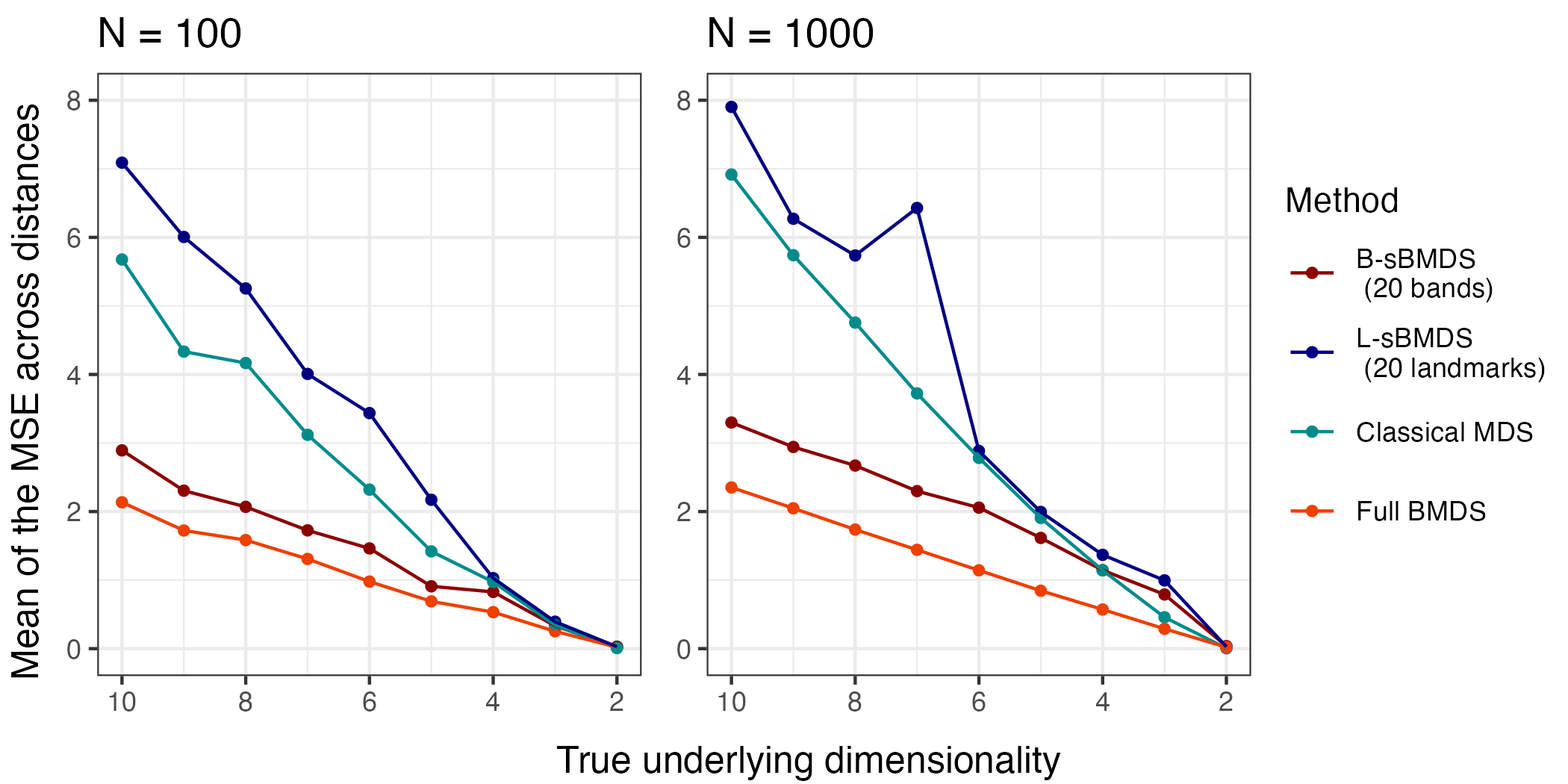}
	\caption{The mean of the mean squared error (MSE) across all distances using $20$ bands/landmarks for 100 and 1,000 data points. We vary the dimension space of the ``true" latent locations while fixing the latent dimensionality to two. We estimate Euclidean distances from the inferred locations obtained using an adaptive Hamiltonian Monte Carlo algorithm under the B-sBMDS, L-sBMDS and full BMDS frameworks. Additionally, we compare the mean MSE across all distances from the inferred locations using classical MDS.}
 \label{fig:HD.error}
\end{figure}

\begin{figure}
	\centering
	\includegraphics[width= 0.9 \textwidth]{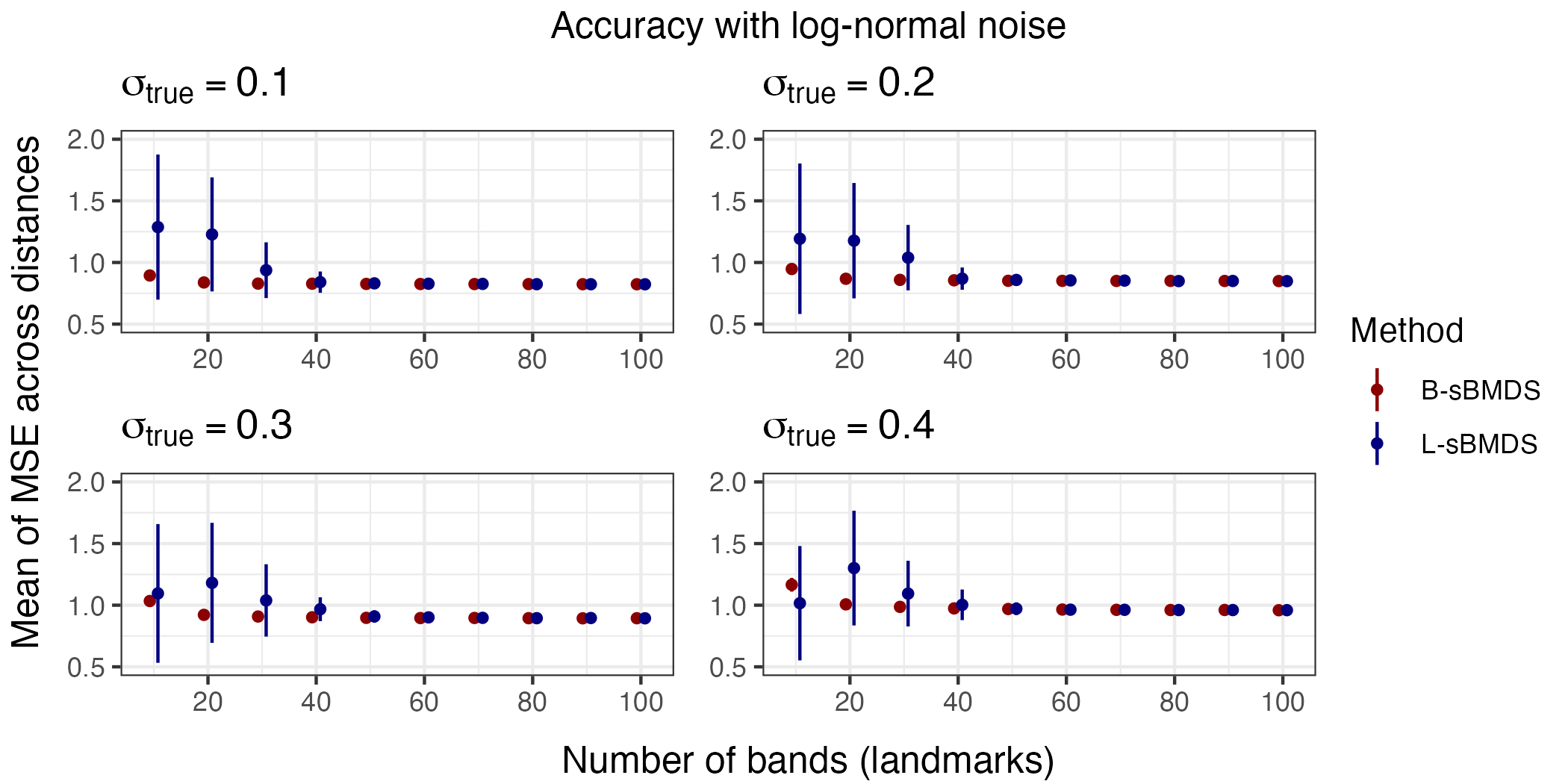}
	\caption{The average mean of the mean squared error (MSE) across all distances from 100 data points evaluated at intervals of 10 bands/landmarks, from 10 to 100, repeated 100 times. The dot is the average mean of MSE, and the error bars are $\pm$ one standard deviation away from this mean. We estimate Euclidean distances from the inferred locations obtained using an adaptive Hamiltonian Monte Carlo algorithm under both sparse Bayesian multidimensional scaling (sBMDS) variants, banded sBMDS (B-sBMDS) and landmark sBMDS (L-sBMDS). When the number of bands (landmarks) equals 100, we return to full BMDS. $\sigma^2_{true}$ is the variance component of the log-normal noise centered at $0$ added to the ``true" distance matrix such that the distribution of the observed distance matrix has heavy-tails.}
 \label{fig:LT.error}
\end{figure}

\subsubsection{Computational performance} \label{sec:compperf}

To better understand the computational benefits of the sBMDS variants, we first calculate the log-likelihood and log-likelihood gradient using B-sBMDS and L-sBMDS for a 10,000 by 10,000 Euclidean distance matrix. Recall that the number of couplings decreases per additional band/landmark. As a result, we see a parabolic-like relationship between evaluation time (in seconds) and the number of bands/landmarks (Figure \ref{fig:rawllgradn10000}). If we were to plot the number of couplings vs seconds per evaluation, we would observe linear associations instead. When the number of bands/landmarks is 10,000, we return to the full case. We observe likelihood (gradient) speedups of 457-fold (773-fold), 91-fold (71-fold), 7-fold (10-fold) and 1.3-fold (1.3-fold) for 5, 50, 500 and 5000 bands (landmarks); there appears to be negligible time differences between B-sBMDS and L-sBMDS. Figure \ref{fig:sBMDS.rawtime.band} emphasizes this correspondence between speedups and number of bands, demonstrating the performance gains using a small number of bands relative to the number of objects. We observe approximately 3-fold, 10-fold and 40-fold speedups when applying the sBMDS likelihoods and gradients to 500, 1,000 and 5,000 data points with 50 bands/landmarks. We only scale up to 50 bands because these are reasonable band counts to achieve high accuracy (Figure \ref{fig:mse_BAND} and \ref{fig:mse_n10000}). We see similar patterns for landmarks in Figure \ref{fig:sBMDS.rawtime.lm} (Appendix \ref{sec:Appendix B}). 

\begin{figure}
	\centering
	\includegraphics[width= 0.9 \textwidth]{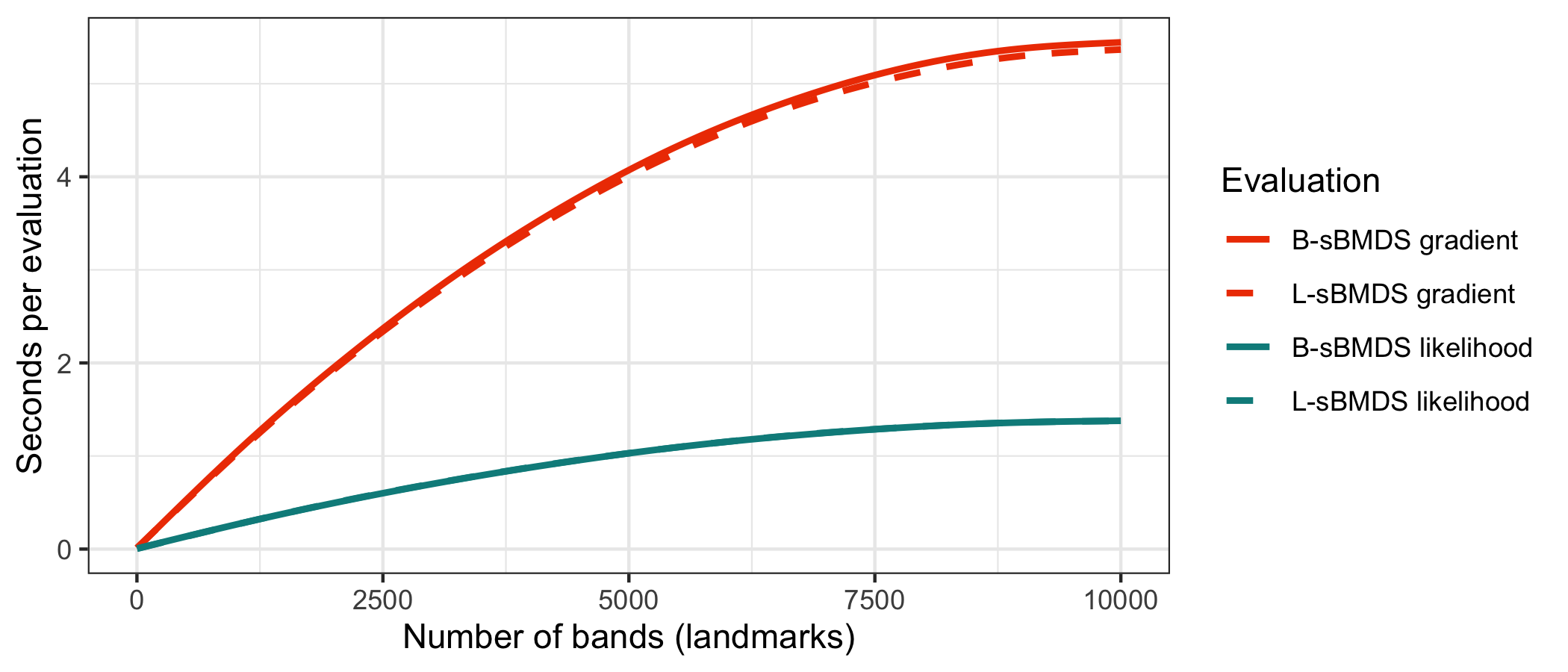}
	\caption{Time elapsed to calculate the sparse BMDS (B-sBMDS and L-sBMDS) likelihoods (cyan) and gradients (red) as a function of the number of bands/landmarks when the number of data points is $10{,}000$. The seconds per evaluation at $10{,}000$ bands/landmarks correspond to the time it takes to calculate the full BMDS likelihoods and gradients. The parabolic curve is due to the number of couplings decreasing per additional band/landmark, causing the differences in computational time to reduce as well. If we plot the number of couplings vs seconds per evaluation, we would observe strictly linear associations.}
 \label{fig:rawllgradn10000}
\end{figure}

\begin{figure}
	\centering
	\includegraphics[width= 0.9 \textwidth]{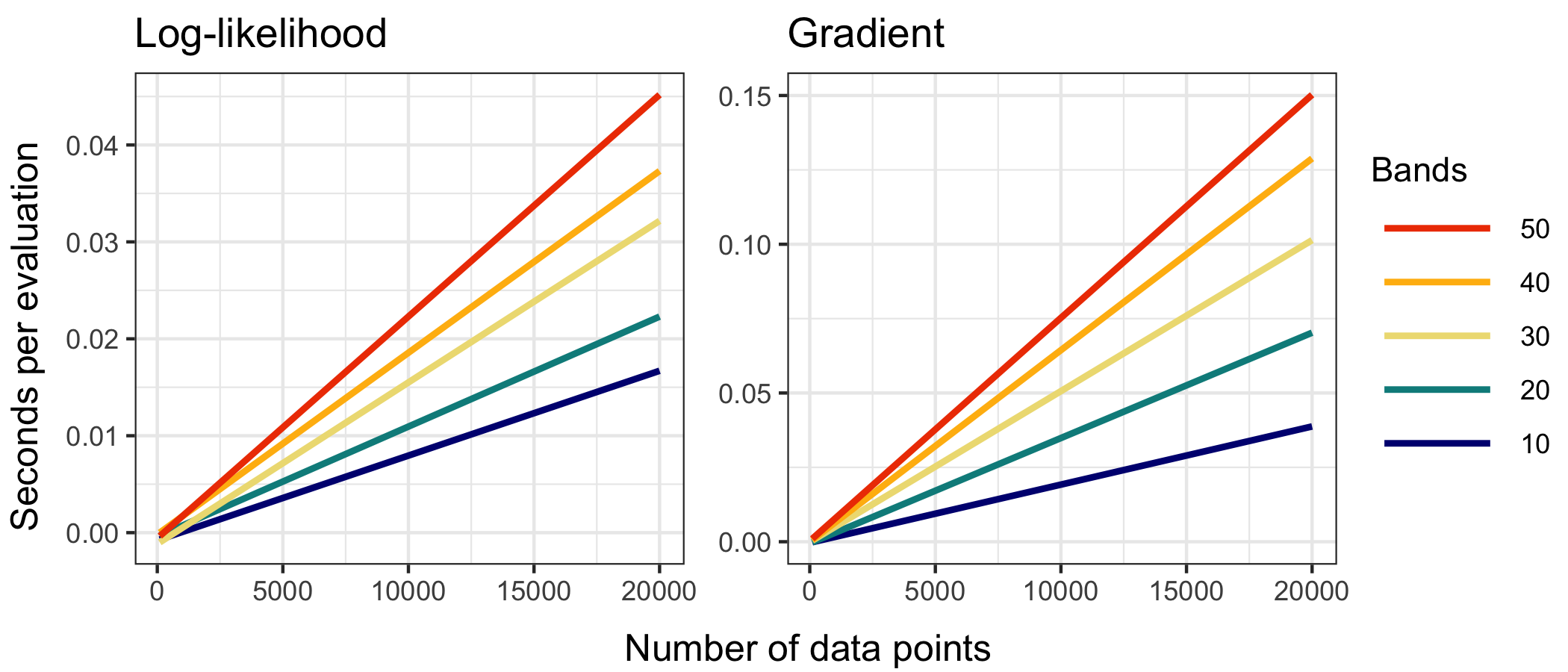}
	\caption{Time elapsed to calculate the banded sparse Bayesian multidimensional scaling (B-sBMDS) likelihood and gradient using $B$ bands as a function of the number of data points.}
 \label{fig:sBMDS.rawtime.band}
\end{figure}

To compare computational performances, we set $\sigma_{true} = 0.2$, a value that will allow us to establish accurate results while obtaining high acceptance probabilities. We fix the number of bands/landmarks to 10 based on the findings from both Figure \ref{fig:mse_BAND} and \ref{fig:mse_LM}, which confirm that this number ensures high model accuracy when $\sigma_{true} = 0.2$ and $N < 1{,}000$. We then conduct MH and HMC under the full BMDS, B-sBMDS, and L-sBMDS models. For a fair comparison, we run all chains until the minimum effective sample size (ESS) is at least 100. ESS is a function of asymptotic auto-correlation, $\text{ESS} = \frac{S}{1 + 2\sum_{t=1}^{\infty} \rho_t}$, where $\rho_t$ is the autocorrelation between samples separated by a lag of $t$ timesteps and $S$ is the length of a time series input. We calculate ESS using the \texttt{coda} package \citep{coda} in \texttt{R}. We define efficiency as the minimum ESS per hour and take the natural log of it to allow comparison across scales. Figure \ref{fig:compperf.miness} compares efficiency across the three models and two MCMC algorithms. The sBMDS variants under HMC outperform the others even in moderately high dimensions. MH begins to break down as the number of data points increases because, while it is computationally faster than HMC, the large dimension of the state space prevents efficient exploration, leading to high auto-correlation and low ESS values. 

\begin{figure}
	\centering
	\includegraphics[width= 0.9 \textwidth]{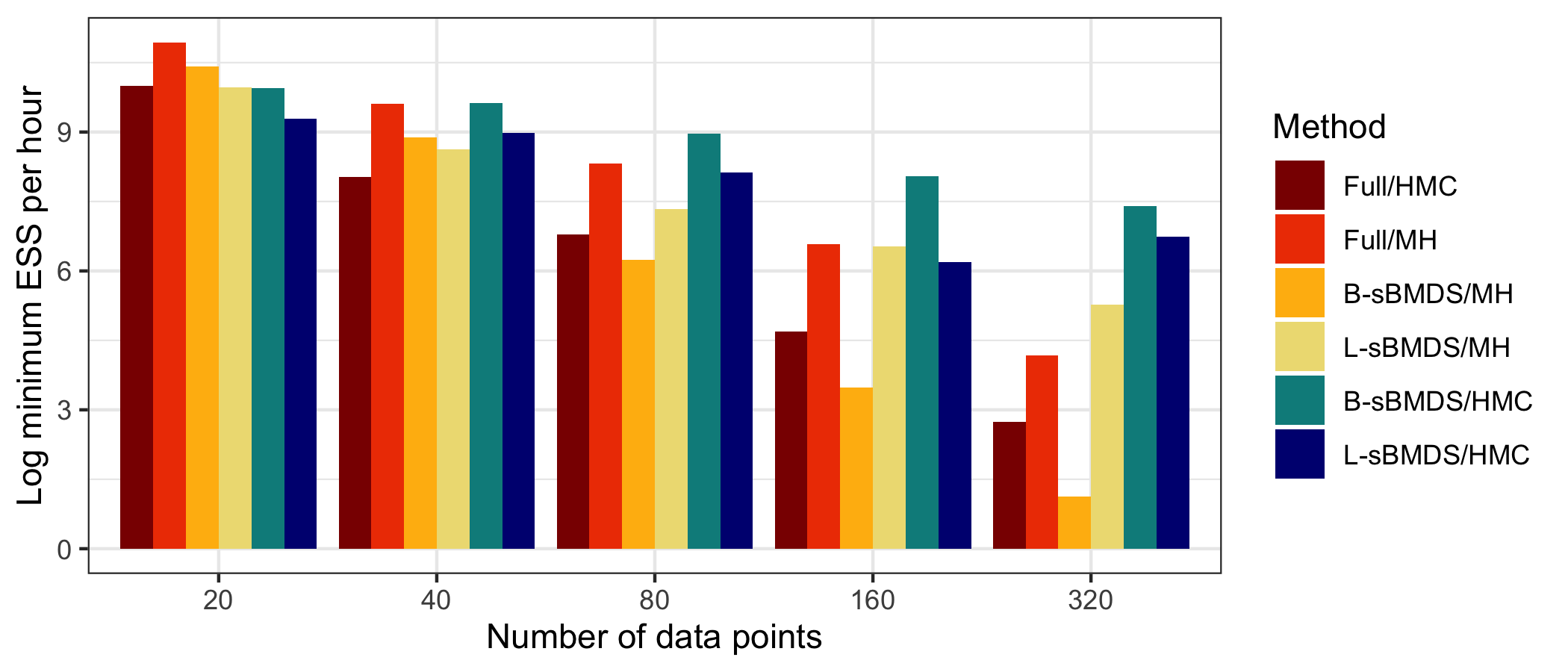}
	\caption{Computational performance measured as the logarithm of the minimum effective sample size (ESS) per hour across different frameworks and number of data point. In the legend, the half before the slash corresponds to the model type. ``Full" is the BMDS model; ``B-sBMDS" (banded sBMDS) and ``L-sBMDS" (landmark sBMDS) are the sparse models using 10 bands/landmarks. The latter half explains the MCMC algorithm used for posterior inference. HMC abbreviates for Hamiltonian Monte Carlo and MH for Metropolis-Hastings.}
 \label{fig:compperf.miness}
\end{figure}

\subsection{Analysis of global influenza} \label{sec:fluapp}
When incorporated into a larger Bayesian hierarchical model, sBMDS provides dimension reduction that propagates and accounts for total model uncertainty. To demonstrate the robustness and applicability of our sparse frameworks, we integrate them within a Bayesian hierarchical model for analyzing the global spread of influenza. Every year seasonal influenza affects millions of adults, resulting in about 140,000 to 710,000 influenza‐related hospitalizations in the United States alone \citep{flu2018}. The virus's ability to constantly evolve makes understanding its viral patterns so important for managing prevalence. The use of easily accessible mobility data can improve the readiness in which we learn about viral epidemics. \citet{holbrook2021bigbmds} apply the BMDS framework to a phylogeographic analysis of the spread of influenza subtypes through transportation networks. They analyze 1,370, 1,389, 1,393 and 1,240 samples of type H1N1, H3N1, Victoria (VIC) and Yamagata (YAM), spanning 12.9, 14.2, 15.4 and 17.75 years, respectively. To scale BMDS to data of this size, they implement core model likelihood and log-likelihood gradient calculations on large graphics processing units and multi-core central processing units. Unfortunately, such an approach requires time-intensive coding and access to expensive computational hardware. We employ a similar Bayesian hierarchical model, applying the same highly structured stochastic process priors but use sBMDS to transform to a latent network space. We are interested in whether under sBMDS we can accurately and efficiently infer the subtype-specific rates of dispersal across the latent airspace for the four influenza strains. 

Our data consists of pairwise ``effective distances" \citep{brockhelbing2013} between countries, which inversely measures the probability of traveling between airports. More trafficked airports have a shorter ``distance" and thus a higher chance of disease transmission. Effective distances are better at predicting disease arrival times and spread compared to geographical distances because they incorporate the underlying mobility network \citep{brockhelbing2013}. For each influenza subtype, we apply sBMDS to their air traffic data with the following priors for the unknown parameters, $\latentData_v, \sigma_v^2$, and hyperparameter $\traitVariance_v$. For strain $v$, the prior on the viral latent locations $\latentData_v$ follows a multivariate Brownian diffusion process along the tree
\begin{align}
    \latentData_v \sim MN(\boldsymbol{\mu}_v, \vec{V}_{\tree_v}, \traitVariance_v),
\end{align}
in which $\boldsymbol{\mu}_v$ is the $N \times D$ mean matrix, $\vec{V}_{\tree_v}$ is the $N \times N$ row covariance matrix calculated from a fixed tree $\tree_v$, and $\traitVariance_v$ is the $D \times D$ column covariance matrix, independently for $v = 1,...,4$. For viral diffusion, $\traitVariance_v$ describes how the virus's location in geographic space covary over lineages. In addition, we assume a priori
\begin{align}
    \traitVariance^{-1}_v \sim Wishart(d_0, \mb{T}_0) \\
    \sigma^{-2}_v \sim Gamma(1, 1).
\end{align}
$d_0$ is the degree of freedom set as the dimension of the latent space and $\mb{T}_0$ is the rate matrix fixed as $\mb{I}_D$ in our model. The trace of $\traitVariance_v$ provides the instantaneous rate of diffusion and is of chief scientific interest. One can think of spatial variance as how much the virus diffuses in a geographic dimension, so by summing up the variance in each dimension, we can understand the total spread of a virus across space in a given moment. We want to accurately infer the trace of $\traitVariance_v$ with our phylogenetic sBMDS model trained on a latent airspace. We implement the adaptive HMC algorithm to recover the viral latent locations along with adaptive MH updates on the BMDS precision parameter, $1 / \sigma^2_v$, and Gibbs updates on $\traitVariance^{-1}_v$. We let the latent dimension be six as \cite{holbrook2021bigbmds} recommended from 5-fold cross-validation. We find 20 leapfrog steps to be adequate as we vary the number of bands/landmarks to 50, 100 and 200.

\subsubsection{Accuracy}
For each subtype and model, we run 120,000 iterations, burning the first 20,000 and saving every 100th iteration. Figure \ref{fig:fluSED} plots the posterior distributions of the strain-specific diffusion rates inferred from the full (left) and banded sparse (right) model. We successfully capture the relative distributions for the B-sBMDS using 50 bands, but note that the posterior modes are slightly off. When we increase the number of bands to 200 (Figure \ref{fig:fluSEDextra}), the distributions appear identical. For Figure \ref{fig:fluMDS}, we apply sBMDS with 50 bands on the H1N1 air traffic data and use procrustes to align the inferred latent locations for each country across all iterations. Since our data has multiple taxon IDs per country, we take the median of the procrustes aligned means of the inferred latent locations and plot the first two dimensions. Figure \ref{fig:fluMDS} demonstrates that we obtain a reasonable map; countries in the same continent group together, and within continents, countries with more air traffic are more centrally located. Using the \texttt{textmineR} package \citep{helliger} in \texttt{R}, we compute the Hellinger distance between the strain-specific posterior distributions of the squared effective distance per year from the full and sparse methods (Table \ref{tab:fluAcc}). As expected, the Hellinger distance decreases with more bands. 

\begin{figure}
	\centering
	\includegraphics[width= 0.9 \textwidth]{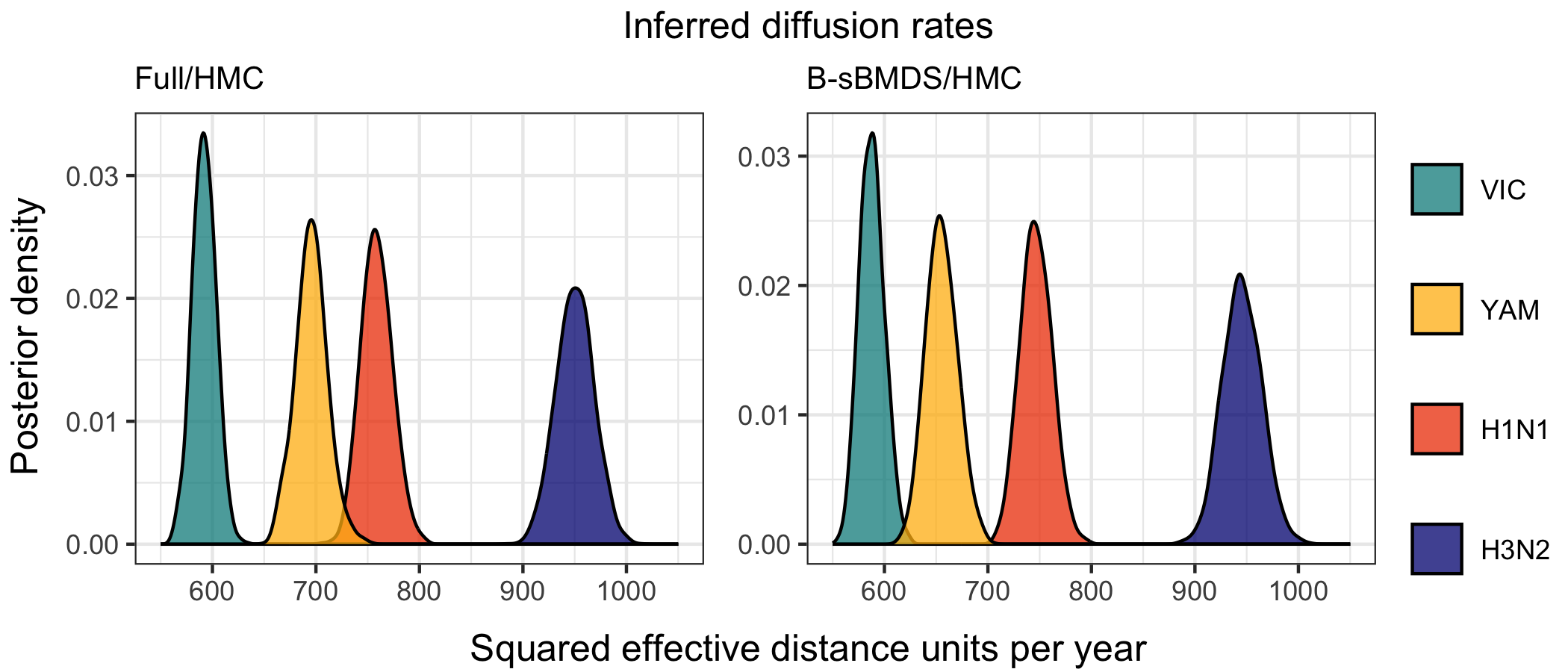}  
        \caption{Posterior distribution of strain-specific diffusion rates inferred from 6-dimensional Bayesian phylogenetic multidimensional scaling with effective world-wide air traffic space distances for data. Full/HMC refers to the use of the full likelihood and gradient whereas B-sBMDS/HMC uses 50 bands to compute the sparse banded likelihood and gradient for inference within the Hamiltonian Monte Carlo algorithm.}
        \label{fig:fluSED}
\end{figure}

\begin{figure}
	\centering
	\includegraphics[width= 0.9 \textwidth]{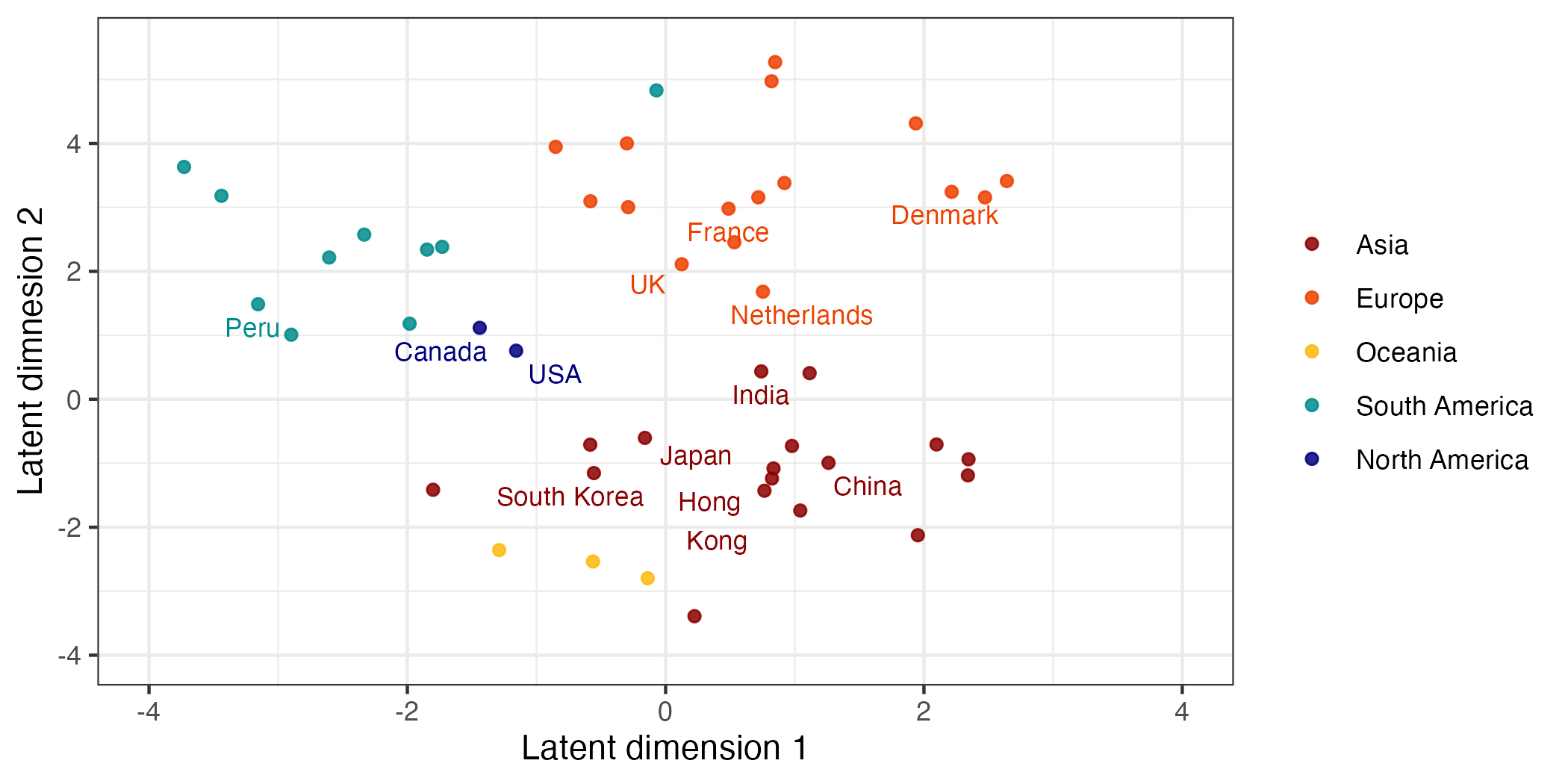}  
        \caption{The first two dimensions of the inferred latent locations for each country using a 6-dimensional sparse Bayesian phylogenetic multidimensional scaling with 50 bands. The plotted inferred latent locations are the median of the Procrustes aligned means across all iterations and taxon.}
        \label{fig:fluMDS}
\end{figure}

\subsubsection{Computational performance}
We measure efficiency speedups across the four influenza subtypes as the ratio of ESS per hour between the full and sparse versions. From Table \ref{tab:fluAcc}, we generally observe that B-sBMDS is more efficient than L-sBMDS, which matches our previous findings (Figure \ref{fig:compperf.miness}). The efficiency speedup decreases with more bands, but is still three times faster for a more than sufficient band count of 200.

\begin{table} [H]
\begin{center}
	\begin{tabular}{c | c | c }
		\multicolumn{3}{c}{B-sBMDS} \\ 
		 B & Hellinger distance & Average efficiency speedup (min, max)\\
		\hline
		50 & 0.024 &  5.99 (5.58, 6.52)\\ 
		100 & 0.021 & 4.06 (3.99, 4.14)\\  
		200 & 0.019 & 2.81 (2.76, 2.86)\\
	\end{tabular}
\quad
	\begin{tabular}{c | c | c }
		\multicolumn{3}{c}{L-sBMDS} \\ 
            L & Hellinger distance & Average efficiency speedup (min, max)\\
		\hline
		50 & 0.024 &  5.22 (4.35, 5.63)\\ 
		100 & 0.023 & 3.83 (3.69, 3.90)\\  
		200 & 0.022 & 2.97 (2.55, 3.52)\\
	\end{tabular}
\end{center}
\caption{We compare the strain-specific posterior distributions of the inferred diffusion rates from the full and sparse BMDS methods. We calculate Hellinger distance between the posterior densities obtained using sparse Bayesian multidimensional scaling (sBMDS) and BMDS. Efficiency speedup is the ratio of effective sample size per hour between the full and sparse BMDS versions. We take the average efficiency speedup across the four influenza subtypes.}
\label{tab:fluAcc}
\end{table}

\subsection{Cluster analysis of ArXiv articles} \label{sec:ERapp}

We explore the utility of sBMDS for propagating uncertainty in downstream tasks by applying a cluster analysis to a collection of ArXiv paper abstracts. Using the \texttt{arxivscraper} package \citep{arXivscrape} in \texttt{Python}, we scrape articles posted on ArXiv from December 2017 to March 2024 across four subject areas: mathematical logic (math.LO), applied physics (physics.app-ph), machine learning (stat.ML) and economics (q-fin.ec). Our final dataset includes 9,308 articles with 1,411 related to math, 1,838 to physics, 5,361 to statistics and 698 to economics. We extract each paper's abstract and embed it into a 768-dimensional numerical vector using \texttt{SentenceTransformers} \citep{sentence-trans} under the \texttt{all-mpnet-base-v2} model \citep{song2020mpnet} in \texttt{Python}. This large language model produces sentence-level embeddings that capture semantic similarity; therefore, abstracts with similar content yield similar embeddings. To form an observed dissimilarity matrix, we compute the pairwise cosine dissimilarities between embeddings, e.g, $1 - \cos{(\latentdata_{n}, \latentdata_{n'})} = 1 - \frac{\latentdata_{n} \cdot \latentdata_{n'}}{||\latentdata_{n}|| ||\latentdata_{n'}||}$ for $\latentdata_{n}, \latentdata_{n'} \in \RR^D$. Cosine dissimilarity is appropriate for textual data, as it emphasizes directional similarity over magnitude.

We apply B-sBMDS to this 9,308 by 9,308 observed distance matrix using 50, 100, 500, 1,000 and 2,000 bands and obtain the posterior distribution of the latent locations over a 2-dimensional space. We use an adaptive HMC algorithm with 20 leapfrog steps to recover the latent locations. We keep every 100th iteration and run enough iterations such that the ESS approximates the number of thinned samples, indicating ``near-independence" among samples. To incorporate the uncertainty encoded in the sBMDS posterior samples, we implement a ``bagged estimator"-style algorithm. We randomly draw $S$ iterations from the joint posterior distribution of the latent locations and implement hierarchical density-based spatial clustering of applications with noise (HDBSCAN) for each iteration using the \texttt{dbscan} package \citep{dbscan} in \texttt{R}. DBSCAN \citep{dbscan1996} is a non-parametric clustering algorithm that groups points into dense regions based on a user-defined radius parameter $\epsilon$ and a minimum number of neighbors. Points within a dense region are assigned to the same cluster, while points in low-density regions are labeled as noise. HDBSCAN \citep{hdbscan2015} performs DBSCAN for various $\epsilon$ values and integrates the results to give the most stable output.

\begin{figure}
    \centering
    \includegraphics[width = 0.9 \textwidth]{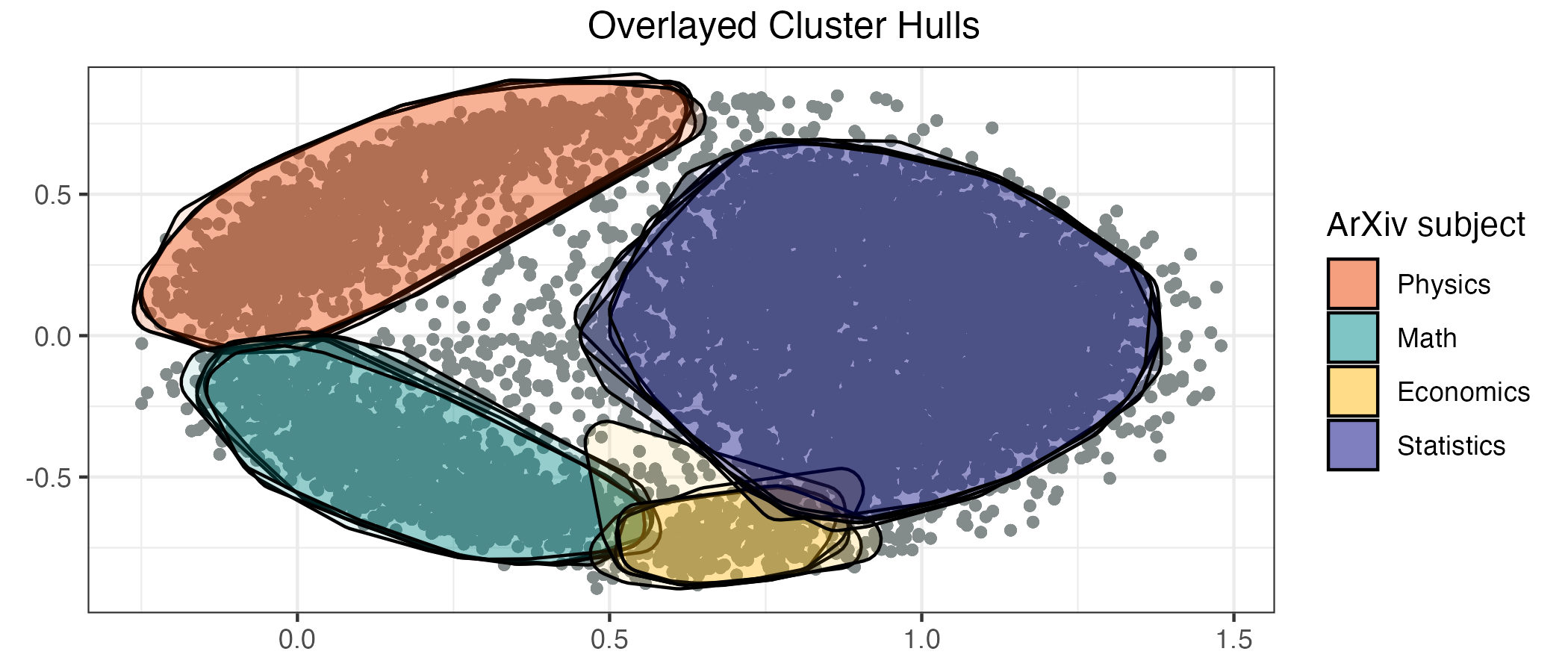}
    \caption{Convex cluster hulls for five posterior samples of the inferred latent locations using hierarchical density-based spatial clustering (HDBSCAN). The inferred locations are estimated from a 2-dimensional B-sBMDS framework using 500 bands and an adaptive Hamiltonian Monte Carlo (HMC) algorithm. The convex hulls are plotted on the posterior mean of the inferred locations across 1,000 iterations, and points outside a hull are considered noise. The ArXiv subject of a cluster is assumed by the most frequent ArXiv category among its members.}
    \label{fig:ER_HP}
\end{figure}

Figure \ref{fig:ER_HP} visualizes the results of HDBSCAN applied to five posterior samples of the latent space plotted on the posterior mean of the latent locations. We set the minimum number of neighbors to 40. Table \ref{tab:ER_accuracy} presents clustering accuracy across band settings. Each cluster is assigned a subject label based on the most frequent ArXiv category among its members. Classification error is the proportion of articles whose ArXiv subject label does not match the subject label of their assigned cluster. We calculate the mean classification error across clusters and note that 500 bands provide a good balance between accuracy and computational cost.

\begin{table}
    \centering
    \scalebox{0.9}{
    \begin{tabular}{c|c|c|c}
         Band & Mode number of clusters & Mean number of noise points & Mean classification error (95\% CI) \\
         \hline
         50 & 2 & 3975 & 0.226 (0.111, 0.315)\\
         100 & 3 & 2954 & 0.094 (0.038, 0.171)\\
         500 & 4 & 1380 & 0.051 (0.041, 0.061)\\
         1000 & 4 & 1248 & 0.038 (0.031, 0.045)\\
         2000 & 4 & 1009 & 0.034 (0.031, 0.037)\\
    \end{tabular}}
    \caption{Accuracy as the number of bands increases using hierarchical density-based spatial clustering (HDBSCAN) on 200 randomly sampled inferred latent locations. The inferred locations are estimated from a 2-dimensional B-sBMDS framework and an adaptive Hamiltonian Monte Carlo (HMC) algorithm. Note that the number of clusters should be 4. HDBSCAN is a non-parametric clustering algorithm and therefore the number of clusters is predicted. Classification error is the proportion of articles whose ArXiv subject label does not match the majority subject label of their assigned cluster, and CI is the credible interval.}
    \label{tab:ER_accuracy}
\end{table}

We obtain clustering labels for each article indirectly induced by sBMDS's uncertainty. Across the $S$ iterations, we compute a co-clustering matrix that estimates the posterior pairwise probability that two articles belong in the same cluster. This matrix can be used for subsequent analyses such as consensus clustering. The combination of sBMDS with HDBSCAN provides an example of how to propagate uncertainty from a sparse Bayesian embedding model to post-hoc inference.  

\section{Discussion} \label{sec:discussion}

We present two methods for subsetting the observed dissimilarity data: banded sparse BMDS (B-sBMDS) and landmark sparse BMDS (L-sBMDS). We show that both sparse methods obtain accurate results at low band/landmark counts even with noisy data. Moreover, combining HMC with sBMDS proves effective in inferring thousands of latent locations. We successfully integrate the sBMDS variants within a Bayesian hierarchical model and demonstrate that one may propagate uncertainty represented by the sBMDS posterior to downstream modeling tasks. In well-specified cases, we recommend using L-sBMDS, as it achieves higher accuracy with fewer landmarks, leading to greater speedups compared to B-sBMDS. However, in most cases, B-sBMDS is the better choice due to its robustness to model misspecification. For example, in our influenza application, in which the Euclidean assumption is violated, B-sBMDS yields a slightly smaller Hellinger distance than L-sBMDS.

Possible extensions to our work include the use of different noise distributions on the observed dissimilarities. For example, \cite{metricBMDS} employ Bayesian metric MDS, assuming the observed dissimilarities come from log-normal distributions. As these distributions still have $\mathcal{O}(N^2)$ time complexity, the sBMDS could be valuable in improving the computational performance for a wider range of dissimilarity data.

Additionally, many potential theoretical developments remain. For example, it appears that one needs approximately $D \sqrt{N}$ bands where $N$ is the number of objects and $D$ is the embedding space. However, we have no formal proof, only experimental results (Sections \ref{sec:fluapp} and \ref{sec:ERapp}, Figure \ref{fig:band_sens}). Determining the number of desired bands/landmarks is difficult due to its data-dependence. We explain in Section \ref{SecHigherDimTheory} how one could extend Theorem \ref{ThmPostConst}'s proof of posterior consistency to higher dimensions. The biggest limitations are extending Lemma \ref{LemmaGaussianObs} and obtaining estimates with good dependence on dimension $D$. One could also explore treating the coupling matrix $J_{n, N}$ as a random variable that depends on the observed data (and perhaps changes over the run-time of an algorithm). An appealing feature of \cite{casecntrl2012} is that they claim reasonable uncertainty quantification along a truly linear run-time. It seems difficult to formalize such a result with posterior consistency for our sBMDS models as the number of bands (landmarks) grows with the number of objects. We are left with many tantalizing questions: ``by including a data-informed approach to model sparsity, can we achieve a linear run-time and still demonstrate posterior consistency?", ``how should we be measuring consistency?", and ``do the datasets \cite{casecntrl2012} study have any special features that change the rate of convergence for a sBMDS-like model?" 

Lastly, we are interested in further extensions within phylogeography. \cite{holbrook2021bigbmds} and \cite{li2023} select the dimension of the latent diffusion process using cross-validation, which is computationally demanding. Therefore, we want to incorporate a shrinkage prior within the Bayesian phylogenetic MDS framework that penalizes the eigenvalues of the diffusion rate matrix. As long as implementing such a prior does not slow down mixing, this approach may help one learn the latent locations in a faster, more unified manner.  


\appendix 
\section{Proof of Theorem \ref{ThmPostConst}}\label{sec:Appendix A}

Throughout this section, we fix notation as in the statement of Theorem \ref{ThmPostConst}.

\subsection{Consistent Estimates of Absolute Values}

We note that $|x_{n}|$ (but not $x_{n}$ itself) is effectively identifiable given the data $\{\delta_{ny}\}_{y \in G_{n, N}^{(k(n))}}$, and we have the posterior concentration bound:

\begin{lemma}\label{LemmaPost}
Fix some $0 < \alpha < 0.1$ and a sequence $\epsilon_{N} = \ell_{N}^{-0.5 + \alpha}$. Then there exist  constants $c_{1},c_{2},c_{3} > 0$  so that for all $N$ sufficiently large and all $n \in [N]$, $k \in [K]$, we have:
\begin{align}
    \mathbb{P}[p(\{ u \, : \, \min(|u - x_{n}|, |u + x_{n}|) \leq c_{1} \epsilon_{N}\} | \{\delta_{ny}\}_{y \in G_{n, N}^{(k(n))}}) \geq 1 - e^{-c_{2} \ell_{N} \epsilon_{N}^{2}}] \geq 1 - e^{-c_{3} \ell_{N} \epsilon_{N}^{2}}.
\end{align}
\end{lemma}

\begin{proof}
Given $x_{n}$, the data $\{\delta_{ny}\}_{y \in G_{n, N}^{(k(n))}}$ are i.i.d. with distributions being a finite mixture of truncated Gaussians. Denote the density of this distribution by $q_{x_{n}}$, and let $\mathcal{F} = \{q_{u}\}_{u \in I}$ be the associated family of possible distributions.

With $\epsilon_{N}$ as above and this choice of $\mathcal{F}$, for any fixed $0 < c < c_{\text{crit}}$ small enough and all $N > N_{0}$ large enough, the sequence $\{\epsilon_{N}\}$ satisfies Inequality (3.1) of \cite{PostConcWong95} for the collection of likelihoods $\mathcal{F}$. Applying Theorem 1 of \cite{PostConcWong95} (together with the well-known formula for Hellinger distances between Gaussians),  there exist constants $c_{1},c_{2},c_{3}$ so that for all $N$ sufficiently large,
\begin{align} \label{IneqMleConc}
    \mathbb{P}[\sup_{u \, : \, \min(|u - x_{n}|, \, |u + x_{n}|) > c_{1}\epsilon_{N}} \prod_{y \in G_{n, N}^{(k(n))}} \frac{q_{u}(\delta^*_{n, y})}{q_{x_{n}}(\delta^*_{n,y})} \geq e^{-c_{2} \ell_{N} \epsilon_{N}^{2}} ] \leq 4 e^{-c_{3} \ell_{N} \epsilon_{N}^{2}},
\end{align} 
where the outer probability is taken with respect to the distribution of the data  $\{\delta_{ny}\}_{y \in G_{n, N}^{(k(n))}}$ given $x_{n}$. On the other hand, for all $u$ satisfying $|u - x_{n}| < \frac{1}{\ell_{N}^{3}}$ and all $N$ sufficiently large, we have 
\begin{align} \label{IneqDenomTriv}
    \prod_{y \in G_{n, N}^{(k(n))}} \frac{q_{x_{n}}(\delta^*_{n,y})}{q_{u}(\delta^*_{n,y})} \leq 2. 
\end{align} 
Combining Inequalities \eqref{IneqMleConc} and \eqref{IneqDenomTriv} completes the proof (with possibly different values of $c_{1},c_{2},c_{3}$).
\end{proof}

\subsection{Consistent Estimates of Signs}

Fix $n, n' \in [N]$ and associated indices $k(n), k(n') \in [K]$. Fix $J = J_{n, N} \cap J_{n', N} \backslash (G_{n, N}^{(k(n))} \cup G_{n', N}^{(k(n'))})$ satisfying $|J| \geq \ell_{N}$. 

Let $\hat{\Tilde{x}}_{n}$ be the posterior median of the distribution of $|x_{n}|$ given $\{ \delta_{ny}\}_{y \in G_{n,N}^{(k(n))}}$, and similarly for $\hat{\Tilde{x}}_{n'}$. For $j \in J$, define the Bernoulli random variables $Z_{j} = \textbf{1}_{A_{j}}$, where $A_{j}$ is the event:
\begin{equation} \label{EqAjDef}
    A_{j}  = \{ \max(\delta_{nj}, \delta_{n'j}) > | \hat{\Tilde{x}}_{n} - \hat{\Tilde{x}}_{n'}| \}.
\end{equation}

Note that $\hat{\Tilde{x}}_{n}, \hat{\Tilde{x}}_{n'}$ are $\{\delta_{ny}\}_{y \in  G_{n,N}^{(k(n))}} \cup \{\delta_{n'y}\}_{y \in  G_{n',N}^{(k(n'))}}$-measurable, and $(\delta_{nj}, \delta_{n'j})$ are independent of $\{\delta_{ny}\}_{y \in  G_{n,N}^{(k(n))}} \cup \{\delta_{n'y}\}_{y \in  G_{n',N}^{(k(n'))}}$ for each $j \in J$, and finally the collection $\{(\delta_{nj}, \delta_{n'j})\}_{j \in J}$ are independent. Thus, conditional on $\{\delta_{ny}\}_{y \in  G_{n,N}^{(k(n))}} \cup \{\delta_{n'y}\}_{y \in  G_{n',N}^{(k(n'))}},$ the random variables $\{Z_{y}\}_{y \in J}$ are i.i.d. Denote by $r_{n,n'}$ their common parameter. By the same argument as in Lemma \ref{LemmaPost}, we have the posterior concentration bound: 

\begin{lemma}\label{LemmaPost2}
Fix notation $0 < \alpha < 0.1$, $\epsilon_{N} = \ell_{N}^{-0.5 + \alpha}$ and notation as above. Then there exist constants $c_{1},c_{2},c_{3} > 0$ so that, for all $N$ sufficiently large,
\begin{align}
    \mathbb{P}[p(\{ r \, : \, |r-r_{n,n'}| \leq c_{1} \epsilon_{N}\} | \{Z_{y}\}_{y \in J}) \geq 1 - e^{-c_{2} \ell_{N} \epsilon_{N}^{2}}] \geq 1 - e^{-c_{3} \ell_{N} \epsilon_{N}^{2}}.
\end{align}
\end{lemma}

We observe that this will allow us to learn whether $\Tilde{x}_{n}, \Tilde{x}_{n'}$ have the same signs (as long as both are far from 0). More precisely, for $j \in J$, define $Y_{j} = \textbf{1}_{B_{j}}$, where 
\begin{align}
    B_{j} = \{ \max(\delta_{nj}, \delta_{n'j}) > | \hat{\Tilde{x}}_{n} + \hat{\Tilde{x}}_{n'}| \}.
\end{align}
By the same argument as the one immediately following Equation \eqref{EqAjDef}, the $Y_{j}$ are i.i.d. Bernoulli. Denote by $q_{n,n'}$ their common parameter. The following is a direct calculation with Gaussians\footnote{If $\sigma = 0$, we'd just look at the probability that the latent position is in the interval $(-\min(|\hat{\Tilde{x}}_{n}|, |\hat{\Tilde{x}}_{n'}|), \min(|\hat{\Tilde{x}}_{n}|, |\hat{\Tilde{x}}_{n'}|))$, for which this is obvious. Since $\sigma > 0$, a complete calculation needs to add in a few additional cases. These doesn't substantially change the results from the trivial case.}:

\begin{lemma} \label{LemmaGaussianObs}
There exists $C,D > 0$ depending on $\sigma$ so that, for all $N$ sufficiently large, the following implication holds:
\begin{align}
    \{ \min(|\hat{\Tilde{x}}_{n}|, |\hat{\Tilde{x}}_{n'}|) > C \epsilon_{N} \} \Rightarrow \{|r_{n, n'} - q_{n, n'}| > D \, \epsilon_{N}\}.
\end{align} 
\end{lemma}

\subsection{Completing the Proof}

We complete the proof of Theorem \ref{ThmPostConst}.

\begin{proof}
For constants $c_{1},c_{2},c_{3}$ to be determined later, we define events 
\begin{align}
    \mathcal{A}_{N} = \left\{ \forall n \in [N], k \in [K], \, p(\{u \, : \, \min(|u - x_{n}|, |u + x_{n}|) \leq c_{1} \epsilon_{N}\} | \{\delta_{ny}\}_{y \in G_{n,N}^{(k(n))}}) \geq 1 - e^{-c_{2} \ell_{N} \epsilon_{N}^{2}} \right\}
\end{align} 
and
\begin{align}
    \mathcal{B}_{N} = \left\{ \forall n,n' \in [N], \, p(\{ r \, : \, |r-r_{n,n'}| \leq c_{1} \epsilon_{N}\} | \{Z_{y}\}_{y \in J}) \geq 1 - e^{-c_{2} \ell_{N} \epsilon_{N}^{2}} \right\}.
\end{align}
Since we have chosen $\epsilon_{N} = \ell_{N}^{-0.5 + \alpha}$ for some $0 < \alpha < 0.1$, we have that $\ell_{N} \epsilon_{N}^{2} \geq \frac{1}{2} N^{\alpha}$ for all $N$ sufficiently large. Thus, by Lemmas \ref{LemmaPost} and \ref{LemmaPost2}, we know that $\mathcal{A}_{N}$ and $\mathcal{B}_{N}$ occur asymptotically almost surely.

On the event $\mathcal{A}_{N}$, we correctly recover $|x_{n}^{(N)}|$ up to additive error $O(\epsilon_{N})$. We now fix a large constant $C$ and consider two cases:
\begin{enumerate}
    \item When $|x_{n}^{(N)}| \leq C \epsilon_{N}$, recovering $|x_{n}^{(N)}|$ up to additive error $O(\epsilon_{N})$ also means recovering $x_{n}^{(N)}$ up to additive error $O(\epsilon_{N})$.
    \item When $|x_{n}^{(N)}| \geq C \epsilon_{N}$ for fixed $C$ sufficiently large,  Lemma \ref{LemmaGaussianObs} implies that on $\mathcal{B}_{N}$ we also recover the sign of $x_{n}^{(N)}$.
\end{enumerate}
Thus, in either case, we recover $x_{n}^{(N)}$ up to additive error $O(\epsilon_{N})$.
\end{proof}

\subsection{Extending Theorem \ref{ThmPostConst} to Higher Dimensions} \label{SecHigherDimTheory}

It is natural to ask if Theorem \ref{ThmPostConst} holds in higher dimensions. The answer appears to be ``yes," but the only proofs that we are aware of have at least one of the following two substantial flaws: they are noticeably longer or give constants $C$ that scale very poorly with dimension. We give here a quick sketch of a proof that closely mimics our one-dimensional argument. It requires no new ideas, but gives bounds that scale very poorly with respect to dimension.

In our proof of Theorem \ref{ThmPostConst}, we invoke Theorem 1 of \cite{PostConcWong95} twice: once in Lemma \ref{LemmaPost} on the ``single row" $G_{n,N}^{(k(n))}$ to show that we have learned $|x_{n}|$ with high accuracy, and again in Lemma \ref{LemmaPost2} on the ``pair of rows with large intersection" $J$ to show that we have learned the sign of $x_{n}$ (as long as $|x_{n}|$ is sufficiently large). To extend this to a higher dimension $D$, we would invoke Theorem 1 of \cite{PostConcWong95} $(D+1)$ times. On the first invocation, we would show that the posterior distribution of $x_{n}$ concentrates near a $(D-1)$-dimensional set that contains the true point. For $1 \leq d \leq D$, in the $d$'th invocation, we would show that we have learned that $x_{n}$ is on a certain subset of dimension $(D-d)$ with high accuracy by looking at $d$ rows of the matrix. Thus, after $D$ invocations, we would have shown that $x_{n}$ is recoverable up to a set of dimension 0. These calculations are nearly identical to the calculations in the current proof. 

The last invocation would be used to deal with ambiguity on a finite set, as in the one-dimensional case. Most of the required changes would be routine (e.g., in part 2 of Assumption \ref{GraphAssumption}, we would need intersections of (D+1) parts of the partition to support our (D+1) invocations of Theorem 1 of \cite{PostConcWong95}). The biggest change comes in proving the natural analogue to Lemma \ref{LemmaGaussianObs}. This calculation is what describes a quantitative sort of identifiability for the model. To extend our arguments to higher dimensions, we need a result along the lines of: ``the set of latent points that (i) lie in a set of dimension $(D-d+1)$ and (ii) have a given expected distance $r$ will lie in a reasonably nice set of dimension $(D-d)$." Ignoring truncations, one can easily check see that the following is true in dimension $d=1$: for $r > 0$, the set of points $x$ such that $E[\delta_{nn'} | x_{n} = x] = r$ is a sphere. When we allow for truncations in dimension $d=1$, we merely need to slice off part of the sphere, and so Lemma \ref{LemmaGaussianObs} is straightforward. For fixed $d > 1$, the various truncations and conditionings involved in repeatedly using this calculation will result in repeated application of unions, intersections and truncation operations to these spheres. In dimension $d=2$, proving the resulting analogue of Lemma \ref{LemmaGaussianObs} in this way is a straightforward but very messy calculus exercise. Unfortunately, we see no easy way to do quick calculations on the resulting set in arbitrary dimension, and no way at all to obtain estimates with reasonable dependence on $D$.

\section{Additional Plots} \label{sec:Appendix B}

Figures \ref{fig:mse_LM} and \ref{fig:sBMDS.rawtime.lm} are analogous to Figures \ref{fig:mse_BAND} and \ref{fig:sBMDS.rawtime.band} from Section \ref{sec:Results}, but under the sparse model using landmarks (L-sBMDS). Figure \ref{fig:mse_LM} demonstrates that very few landmarks are necessary to achieve high accuracy relative to the number of data points. Figure \ref{fig:sBMDS.rawtime.lm} plots the raw speed-ups, varying the number of landmarks as the number of data points increases. Finally, Figure \ref{fig:fluSEDextra} illustrates the posterior distribution of the strain-specific diffusion rates under the B-sBMDS/HMC model using 100 and 200 bands. When the number of bands is 200, we see no apparent difference from the full BMDS plot in Figure \ref{fig:fluSED}.

\setcounter{figure}{0}
\counterwithin{figure}{section}

\begin{figure}[H]
	\centering
	\includegraphics[width= 0.9 \textwidth]{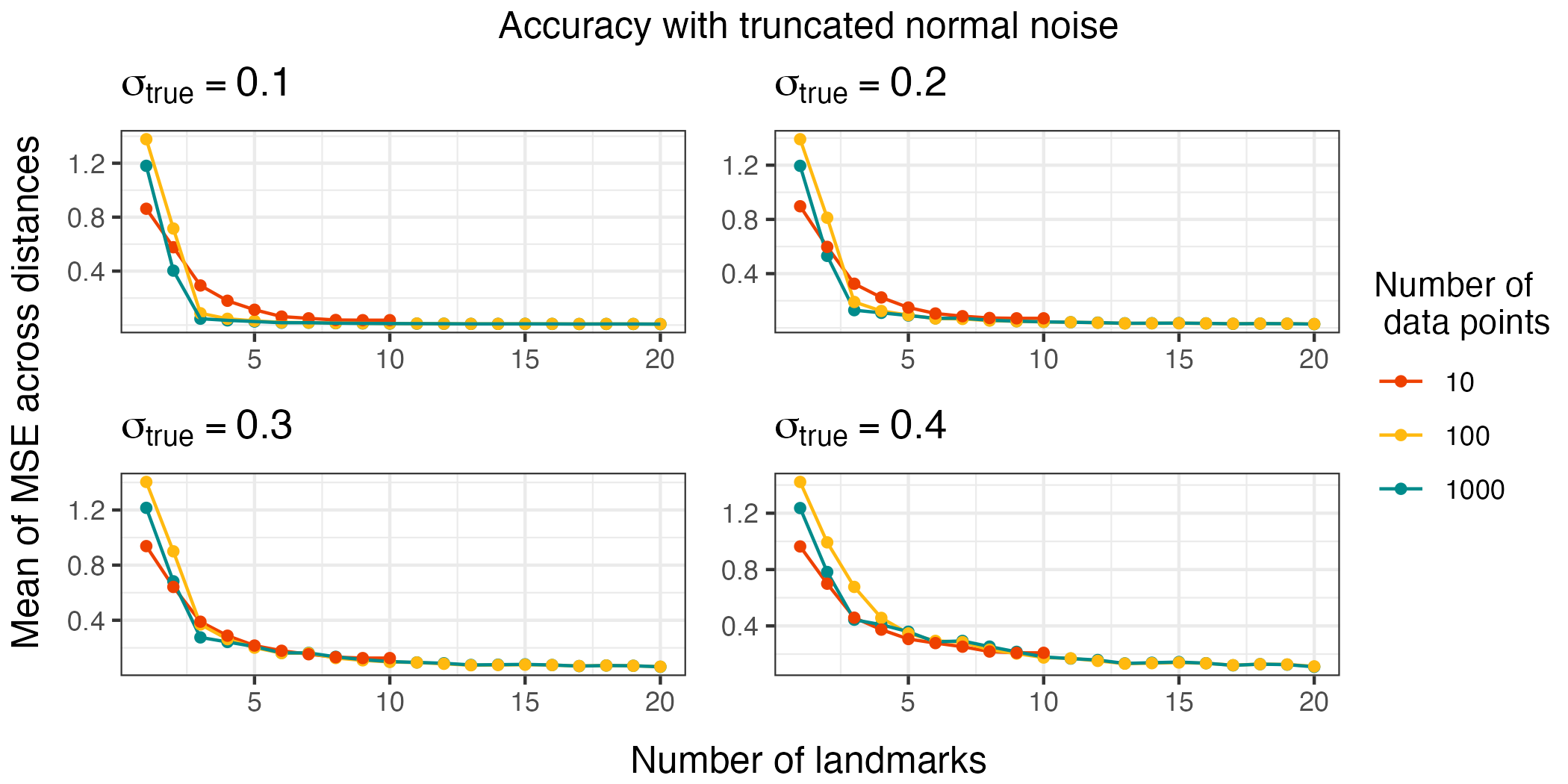}
	\caption{The mean of the mean squared error (MSE) across all distances using $1$ to $10$ landmarks for $10$ data points and $1$ to $20$ landmarks for $100$ and $1{,}000$ data points. We estimate Euclidean distances from the inferred locations obtained using an adaptive Hamiltonian Monte Carlo algorithm under landmark sparse Bayesian multidimensional scaling (L-sBMDS). $\sigma^2_{true}$ is the variance component of the truncated normal noise centered at $0$ added to the ``true" distance matrix such that $\sigma_{true}$ corresponds to the BMDS error standard deviation $\sigma$}
 \label{fig:mse_LM}
\end{figure}

\begin{figure}[H]
	\centering
	\includegraphics[width= 0.9 \textwidth]{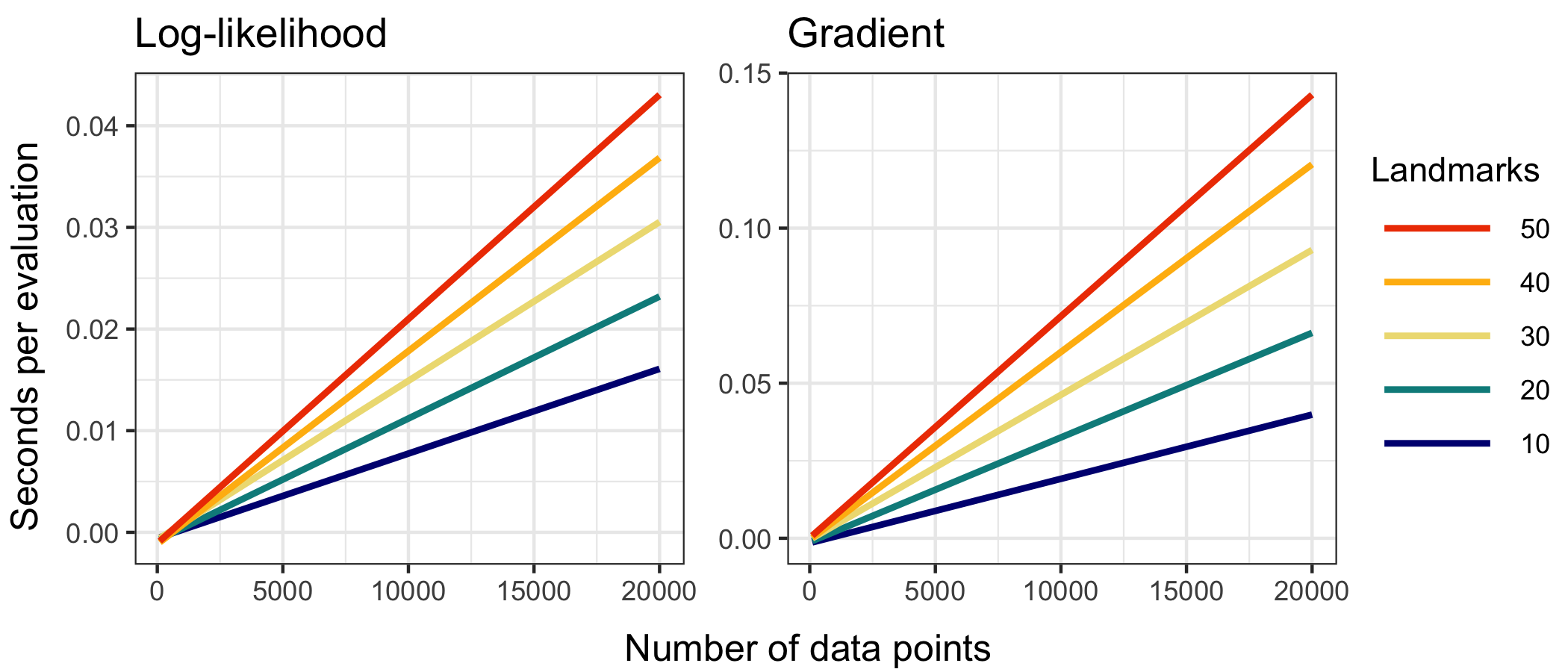}
	\caption{Time elapsed to calculate the landmark sparse Bayesian multidimensional scaling (L-sBMDS) likelihood and gradient using $L$ landmarks as a function of the number of data points.}
 \label{fig:sBMDS.rawtime.lm}
\end{figure}

\begin{figure}[H]
	\centering
	\includegraphics[width= 0.9 \textwidth]{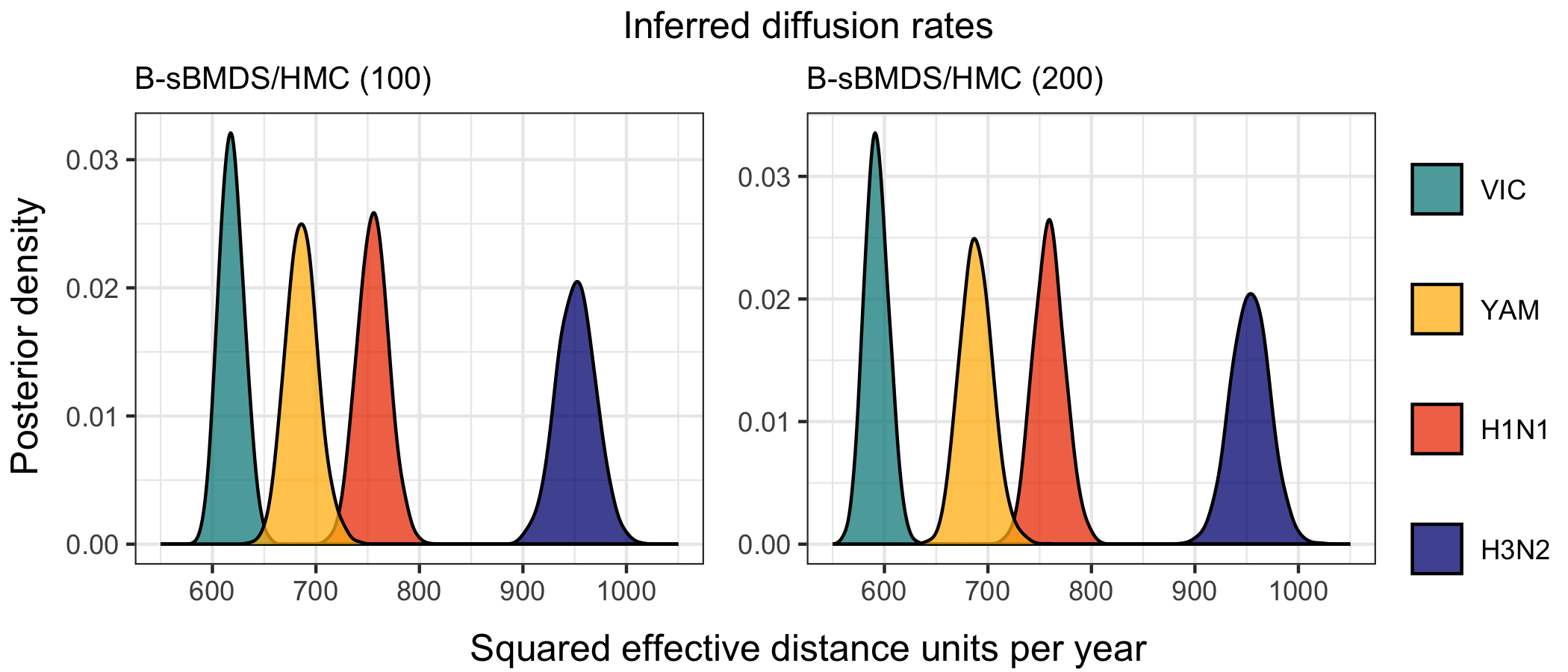}
	\caption{Posterior distribution of strain-specific diffusion rates inferred from 6-dimensional Bayesian phylogenetic multidimensional scaling with effective world-wide air traffic space distances for data. B-sBMDS/HMC uses 100 (left) and 200 (right) bands to compute the sparse banded likelihood and gradient for inference within the Hamiltonian Monte Carlo algorithm.}
 \label{fig:fluSEDextra}
\end{figure}

\begin{figure}[H]
	\centering
	\includegraphics[width= 0.9 \textwidth]{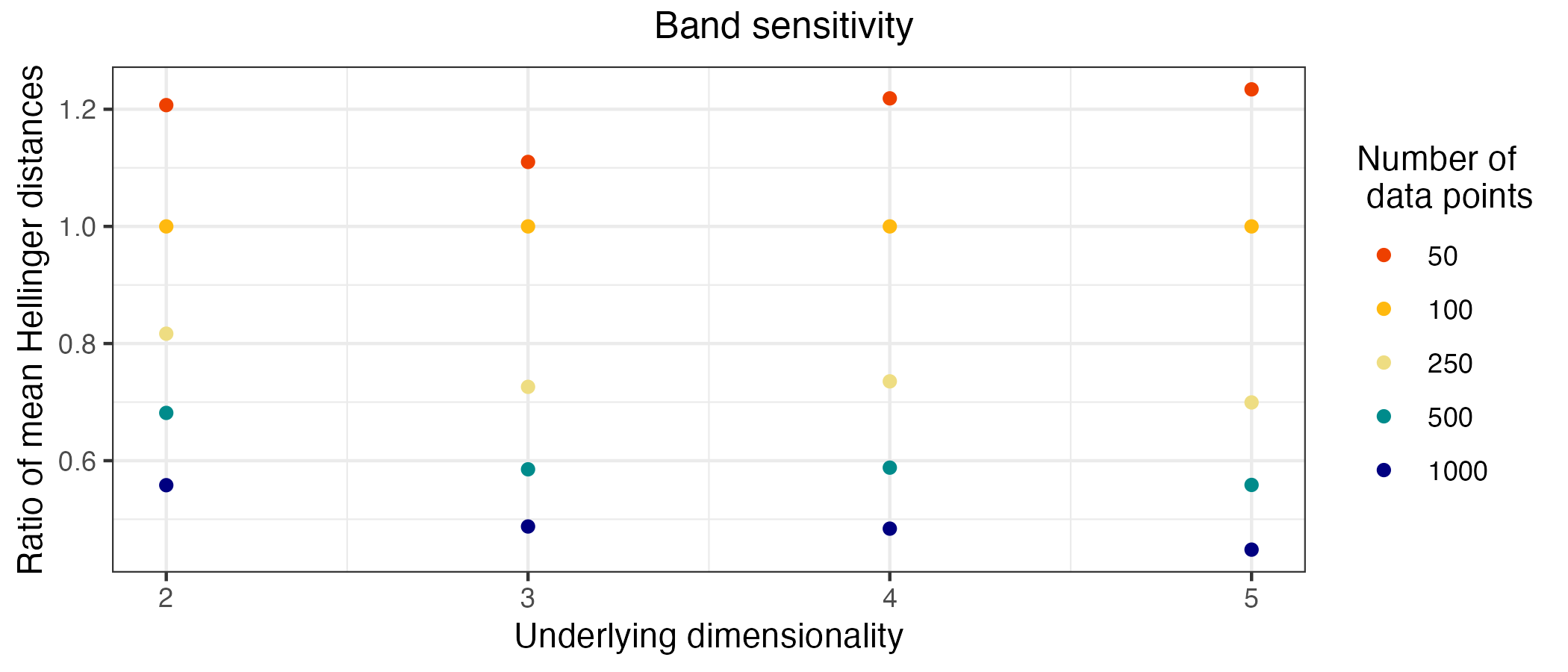}
	\caption{Ratio of the mean Hellinger distance for various dimensionalities ($D$) and number of data points ($N)$. Hellinger distances are computed from the posterior distributions of the estimated Euclidean distances between full BMDS and banded sBMDS with $\lceil D \sqrt{N} \rceil$ bands. Data is generated under Gaussian assumptions as described in Section \ref{sec:Results}. For fixed $D$, we calculate the ratio as the mean Hellinger distance with $N$ data points and with $N = 100$. The near horizontal lines across the underlying dimensionality indicate that the selected number of bands obtains results that are very similar to the full model despite the increase in dimensions. The ratio of the Hellinger distance is smaller for larger $N$ as error goes down at some statistical rate, e.g., of $\frac{1}{\sqrt{N}}$.}
 \label{fig:band_sens}
\end{figure}

\section*{Acknowledgments}

This work was supported by the NIH (K25 AI153816) and the NSF (DMS 2152774 and DMS 2236854).

\bibliography{references}

\end{document}